\documentclass[11pt]{article}
\usepackage{mathtools}
\usepackage{url}
\usepackage[table]{xcolor}
\usepackage{booktabs}
\usepackage{algorithm}
\usepackage{algorithmicx}
\usepackage{algpseudocode}
\usepackage[numbers,sort]{natbib}
\usepackage{graphicx}
\usepackage{caption}
\usepackage{amsmath,amsthm,amssymb}
\usepackage{fullpage}
\usepackage{subcaption}
\usepackage{times}

\DeclareMathOperator*{\argmax}{argmax}

\DeclareMathOperator*{\E}{E}

\newtheorem{theorem}{Theorem}
\newtheorem{corollary}[theorem]{Corollary} 

\newtheorem{example}{Example}

\newcommand{\n}{v}
\newcommand{\ns}{V}
\newcommand{\actsn}[1][\n]{A_{#1}}
\newcommand{\actin}[2][\n]{{#1}^{#2}}
\newcommand{\infls}{\mathcal{F}}
\newcommand{\infln}[1][\n]{f_{#1}}
\newcommand{\tn}[1][\n]{\theta_{#1}}
\newcommand{\ts}{\Theta}
\newcommand{\pset}[1]{2^{#1}}

\newcommand{\wgt}{w}
\newcommand{\obj}{\sigma}

\newcommand{\xs}{\mathbf{x}}
\newcommand{\xn}[1][\n]{x_{#1}}
\newcommand{\ys}{\mathbf{y}}
\newcommand{\yn}[1][\n]{y_{#1}}

\newcommand{\wgtM}{\hat{\wgt}}
\newcommand{\inflsM}{\hat{\infls}}
\newcommand{\inflnM}[1][\nM]{\hat{f}_{#1}}
\newcommand{\nsM}{\hat{\ns}}
\newcommand{\nM}[1][\n]{\hat{#1}}
\newcommand{\objM}{\hat{\obj}}

\newcommand{\disc}{\delta}
\newcommand{\Disc}{\Delta}

\newcommand{\unitn}[1][\n]{\boldsymbol{\disc}_{#1}}

\newcommand{\unif}{\mathcal{U}}
\newcommand{\abs}[1]{\lvert{#1}\rvert}

\newcommand{\inst}{\mathcal{I}}
\newcommand{\instM}{\hat{\inst}}

\newcommand{\eps}{\varepsilon}
\let\epsilon=\varepsilon

\newcommand{\norm}[2][]{{\lVert{#2}\rVert}_{#1}}

{\makeatletter
 \gdef\xxxmark{%
   \expandafter\ifx\csname @mpargs\endcsname\relax 
     \expandafter\ifx\csname @captype\endcsname\relax 
       \marginpar{xxx}
     \else
       xxx 
     \fi
   \else
     xxx 
   \fi}
 \gdef\xxx{\@ifnextchar[\xxx@lab\xxx@nolab}
 \long\gdef\xxx@lab[#1]#2{{\bf [\xxxmark #2 ---{\sc #1}]}}
 \long\gdef\xxx@nolab#1{{\bf [\xxxmark #1]}}
}

\newif\ifabstract
\abstractfalse
\newif\iffull
\ifabstract \fullfalse \else \fulltrue \fi

\ifabstract

\permission{Copyright is held by the International World Wide Web Conference Committee (IW3C2). IW3C2 reserves the right to provide a hyperlink to the author's site if the Material is used in electronic media.}
\conferenceinfo{WWW'14,}{April 7--11, 2014, Seoul, Korea.} 
\copyrightetc{ACM \the\acmcopyr}
\crdata{978-1-4503-2744-2/14/04. \\
http://dx.doi.org/10.1145/2566486.2568039}

\fi

\begin{document}


\title{How to Influence People with Partial Incentives%
\thanks{This work was supported in part by NSF CAREER award 1053605, NSF grant CCF-1161626, ONR YIP award N000141110662, DARPA/AFOSR grant FA9550-12-1-0423.}}

%
\author{
  Erik D. Demaine\thanks{MIT, Cambridge, MA, USA. \tt{edemaine@mit.edu}.}
  \and
  MohammadTaghi Hajiaghayi\thanks{University of Maryland, College Park, MD, USA. \tt{hajiagha@cs.umd.edu}.}
  \and
  Hamid Mahini\thanks{University of Maryland, College Park, MD, USA. \tt{hamid.mahini@gmail.com}.}
  \and  
  David L. Malec\thanks{University of Maryland, College Park, MD, USA. \tt{dmalec@umiacs.umd.edu}.}
  \and
  S. Raghavan\thanks{University of Maryland, College Park, MD, USA. \tt{raghavan@umd.edu}.}
  \and
  Anshul Sawant\thanks{University of Maryland, College Park, MD, USA. \tt{asawant@cs.umd.edu}.}
  \and
  Morteza Zadimoghadam\thanks{MIT, Cambridge, MA, USA. \tt{morteza@csail.mit.edu}}
}

\maketitle
\begin{abstract}
We study the power of fractional allocations of resources
to maximize our influence in a network.  This work extends in a natural way
the well-studied model by Kempe, Kleinberg, and Tardos (2003),
where a designer selects a (small) seed set of nodes in a social network to
influence directly, this influence cascades when other nodes reach
certain thresholds of neighbor influence, and the goal is to maximize
the final number of influenced nodes.
Despite extensive study from both practical and theoretical viewpoints,
this model limits the designer to a binary choice for each node,
with no chance to apply intermediate levels of influence.
This model captures some settings precisely, such as exposure to
an idea or pathogen, but it fails to capture very relevant concerns in others,
for example, a manufacturer promoting a new product by distributing
five ``20\% off'' coupons instead of giving away a single free product.

While fractional versions of problems tend to be easier to solve
than integral versions, for influence maximization, we show that the
two versions have essentially the same computational complexity.
On the other hand, the two versions can have vastly different solutions:
the added flexibility of fractional allocation can lead to significantly improved
influence.  Our main theoretical contribution is to show how to adapt the
major positive results from the integral case to the fractional case.
Specifically,
Mossel and Roch (2006) used the submodularity of influence to obtain their
integral results; we introduce a new notion of \emph{continuous submodularity},
and use this to obtain matching fractional results.
We conclude that we can achieve the same greedy $(1-1/e-\eps)$-approximation for the fractional case as the integral case, and that other heuristics are likely to carry over as well.  In practice, we find that the fractional model
performs substantially better than the integral model, according to
simulations on real-world social network data.

\end{abstract}
\thispagestyle{empty}

\newpage
\setcounter{page}{1}

\section{Introduction}
The ideas we are exposed to and the choices we make are significantly
influenced by our social context.  It has long been studied how our social
network (i.e., who we interact with) impacts the choices we make,
and how ideas and behaviors can spread through social
networks~\cite{Bikhchandani1992, Granovetter1978, Schelling2006, Valente1995}.
With websites such as Facebook and
Google+ devoted to the forming and maintaining of social networks,
this effect becomes ever more evident.  Individuals are linked together
more explicitly and measurably, making it both easier and more important
to understand how social networks affect the behaviors and actions
that spread through a society.

A key problem in this area is to understand how such a
behavioral cascade can start.  For example, if a company wants to
introduce a new product but has a limited promotional budget, it
becomes critical to understand how to target their promotional efforts
in order to generate awareness among as many people as possible.  A
well-studied model for this is the Influence Maximization problem,
introduced by Kempe, Kleinberg, and Tardos \cite{kkt03}.
The problem's objective is to find a small set of individuals to influence,
such that this influence will cascade and grow through the social network
to the maximum extent possible.
For example, if a company wants to introduce a new piece of
software, and believes that friends of users are likely to become users
themselves, how should they allocate free copies of their software in
order to maximize the size of their eventual user base?

Since the introduction of the Influence Maximization problem \cite{kkt03},
there has been a great deal of interest and follow-up work in the model.
While Kempe et al.~\cite{kkt03} give
a greedy algorithm for approximating the Influence Maximization
problem, it requires costly simulation at every step; thus, while
their solution provides a good benchmark, a key area of research has
been on finding practical, fast algorithms that themselves provide
good approximations to the greedy algorithm~\cite{Borgs2012, Chen2010, Chen2009a, Chen2010a, Leskovec2007}.
The practical, applied nature of the motivating settings means that
even small gains in performance (either runtime or approximation
factor) are critical, especially on large, real-world instances.

We believe that the standard formulation of the Influence Maximization
problem, however, misses a critical aspect of practical applications.
In particular, it forces a binary choice upon the optimizer, forcing
a choice of either zero or complete influence on each individual,
with no options in between.  While this is reasonable for some settings ---
e.g., exposure to an idea or pathogen -- it is far less reasonable for
other settings of practical importance.  For example, a company promoting a
new product may find that giving away ten free copies is far less
effective than offering a discount of ten percent to a hundred people.
We propose a \emph{fractional} version of the problem where the
optimizer has the freedom to split influence across individuals
as they see fit.

To make this concrete, consider the following problem an optimizer
might face.  Say that an optimizer feels there is some small,
well-connected group whose adoption of their product is critical to
success, but only has enough promotion budget remaining to influence
one third of the group directly.  In the original version of Influence
Maximization, the optimizer is forced to decide which third of the
group to focus on.  We believe it is more natural to assume they have
the flexibility to try applying uniform influence to the group, say
offering everyone a discount of one third on the price of their
product, or in fact any combination of these two approaches.  While
our results are preliminary, we feel that our proposed model addresses
some very real concerns with practical applications of Influence Maximization,
and offers many opportunities for important future research.


\subsection{Our Results and Techniques}
This work aims to understand how our proposed
fractional version of the Influence Maximization problem differs from
the integral version proposed by Kempe, Kleinberg, and Tardos \cite{kkt03}.
We consider this question from both a theoretical and an empirical perspective.
On the theoretical side, we show that, unlike many problems, the fractional
version appears to retain essentially the same computational hardness as the
integral version.  The problems are very different, however:
we give examples where the objective values for the fractional and integral
versions differ significantly.  Nonetheless, we are able to carry over the
main positive results to the fractional setting, notably that the objective function is submodular and the problem therefore admits a greedy $(1-1/e-\eps)$-approximation.
On the empirical side, we simulate the main algorithms and heuristics on
real-world social network data, and find that the computed solutions are
substantially more efficient in the fractional setting.

Our main theoretical result shows that the positive results of
Mossel and Roch~\cite{mr10} extend to our proposed fractional model.  Their result
states that, in the integral case, when influence between individuals is
submodular, so too is the objective function in Influence
Maximization.  We show that, for a continuous analog of
submodularity,\footnote{Note that our notion of continuous submodularity is
  neither of the two most common continuous extensions of submodularity,
  namely the multilinear and Lov\'{a}sz extensions.}
the same results holds for our fractional
case.  First we consider a discretized version of the fractional
Influence Maximization Problem, where each vertex can be assigned a
weight that is a multiple of some discretization parameter $\epsilon =
\frac{1}{N}$.  Then we consider the final influenced set by choosing a
weighted seed set $S$, where the weight of each element is a multiple
of~$\eps$.
We show that the fractional Influence Maximization objective is a
submodular function of $S$ for any $N\geq 1$ (Theorem~\ref{thm:disc}).
We further extend this result to the fully continuous case
(Theorem~\ref{thm:cont}).  Note that this result does not follow
simply by relating the fractional objective function to the integral
objective and interpolating, or other similar methods; instead, we need to
use a nontrivial reduction to the generalization of the influence
maximization problem given by Mossel and Roch~\cite{mr10}.  Not only does this result
show that
our problem admits a greedy $(1-1/e-\eps)$-appromixation algorithm,  it furthermore gives us hope that we can readily adapt the
large body of work on efficient heuristics for the integral case to
our problem and achieve good results.

In addition to showing the submodularity of the objective persists
from the integral case to the fractional case, we show that the
hardness of the integral case persists as well.  In the case of fixed
thresholds, we show that all of the hardness results of Kempe et al.~\cite{kkt03}
extend readily to the fractional case.
Specifically, we show that, for the fractional version of the linear
influence model, even finding an $n^{1-\epsilon}$-approximation is
NP-hard. First we prove NP-hardness of the problem by
a reduction from  Independent Set (Theorem~\ref{thm:fixed_nphard}),
and then we strengthen the result to prove inapproximability
(Corollary~\ref{cor:apphard}). In addition, when thresholds are
assumed to be independent and uniformly distributed in $[0,1]$, we show
that it is NP-hard to achieve better than a $(1-1/e)$-approximation
in the Triggering model introduced by Kempe et al.~\cite{kkt03}.
This holds even for the simple case where triggering sets are
deterministic and have constant sizes, and shows that even for this
simple case the greedy approximation is tight, just as in the integral
case.  An important aspect of all of these reductions is that they use
very simple directed acyclic graphs (DAGs), with only two or three layers of vertices.

Our last set of results focuses on the special case where the network is
a DAG.  Here, we focus
on the linear influence model with uniform thresholds.  In this case, we see that we
can easily compute the expected influence from any single node via
dynamic programming; this closely resembles a previous result for the
integral case \cite{Chen2010a}.  In the fractional case, this gives us
a sort of linearity result.  Namely, if we are careful to avoid
interference between the influences we place on nodes, we can conclude
that the objective is essentially linear in the seed set.  While the
conditions on this theorem seem strong at first glance, it has a very
powerful implication: all of the hardness results we presented
involved choosing optimal seed sets from among the sources in a DAG,
and this theorem says that with uniform thresholds the greedy
algorithm finds the \emph{optimal} such seed set.


\subsection{Related Work}
Economics, sociology, and political science have all studied and modeled
behaviors arising from information and influence cascades in social networks.
Some of the earliest models were proposed by Granovetter~\cite{Granovetter1978} and
Schelling~\cite{Schelling2006}. Since then, many such models have been studied and
proposed in the literature \cite{Bikhchandani1992, Rogers2010, Valente1995}.

The advent of social networking platforms such as Facebook, Twitter, and
Flickr has provided researchers with unprecedented data about social
interactions, albeit in a virtual setting. The question of monetizing this
data is critically important for the entities that provide these platforms and
the entities that want to leverage this data to engineer effective marketing
campaigns. These two factors have generated huge interest in algorithmic
aspects of these systems.

A question of central importance is to recognize ``important individuals'' in
a social network. Domingos and Richardson~\cite{Domingos2001, Richardson2002}
were the first to propose heuristics for selection of customers on a network
for marketing. This work focuses on evaluating customers based on their
intrinsic and network value. The network value is assumed to be generated by a
customer influencing other customers in her social network to buy the product.
In a seminal paper, Kempe et al.~\cite{kkt03} give an approximation algorithm for
selection of influential nodes under the linear threshold (LT) model.
Mossel and Roch~\cite{Mossel2007} generalized the results of Kempe et al.~\cite{kkt03} to cases where the
activation functions are monotone and submodular. Gunnec and
Raghavan~\cite{GRpap1} were the first to discuss fractional incentives (they
refer to these as partial incentives/inducements) in the context of a product
design problem. They consider a fractional version of the target set selection
problem (i.e., fixed thresholds, fractional incentives, a linear influence
model, with the goal of minimizing the fractional incentives paid out so that
all nodes in the graph are influenced). They provide an integer programming
model, and show that when the neighbors of a node have equal influence on it,
the problem is polynomially solvable via a greedy algorithm \cite{GRpap1,
Gunnecthesis, GRZpap2}.
 
Some recent work has directly tackled the question of revenue maximization in
social networks by leveraging differential pricing to monetize positive
externalities arising due to adoption of a product by neighbors of a
customer~\cite{akhlaghpour2010optimal, arthur2009pricing, Ehsani2012118,
hartline2008optimal}.  The closest such work is by Singer~\cite{Singer12}, but it still restricts the planner's direct influence to initial adoption. Other work has focused on finding faster algorithms for
the target set selection problem~\cite{Chen2010, Chen2009a, Chen2010a,
Leskovec2007}. A very recent theoretical result in this direction is an
$O(\frac{(m+n)\log n}{\epsilon^3})$ algorithm giving an approximation
guarantee of $1 - \frac{1}{e} - \epsilon$~\cite{Borgs2012}. While
Leskovec et al.~\cite{Leskovec2007} do not compare their algorithm directly with the greedy
algorithm of Kempe et al.~\cite{kkt03}, the heuristics in other papers \cite{Chen2010,
Chen2009a, Chen2010a} approach the performance of the greedy algorithm quite
closely. For example, in \cite{Chen2010}, the proposed heuristic achieves an
influence spread of approximately 95\% of the influence spread achieved by the
greedy algorithm. An interesting fact on the flip side is that none of the
heuristics beat the greedy algorithm (which itself is a heuristic) for even a
single data set.


\section{Model}
\label{sec:model}
\newcommand{\ws}{{\mathbf w}}
\newcommand{\w}[1]{{w}_{#1}}
\newcommand{\nbdin}{{\delta^{^{-}}}}
\newcommand{\nbdout}{{\delta^{^+}}}
\newcommand{\Infl}{I}
\newcommand{\Sat}{S}

\paragraph*{\bf Integral Influence Model} 
We begin by describing the model used for propagation of influence in
social networks used by Mossel and Roch~\cite{mr10}; it captures the model of Kempe et al.~\cite{kkt03} as a special case.  While the latter described the spread of influence
in terms of an explicit network, the former leaves the underlying
social network implicit.  In this model, the social network is given by a vertex set
$\ns$ and an explicit description of how vertices influence each
other. For each vertex $\n\in\ns$, we are given a
function $\infln:\pset{\ns}\rightarrow[0,1]$ specifying the amount of
influence each subset $S\subseteq\ns$ exerts on $v$.  We denote the
set of all influence functions by $\infls=\{\infln\}_{\n\in\ns}$.

Given a social network specified by $(\ns,\infls)$, we want to understand how
influence propagates in this network.  The spread of influence is
modeled by a process that runs in discrete stages.  In addition to the
influence function $\infln$, each vertex $\n$ has a threshold
$\tn\in[0,1]$ representing how resistant it is to being influenced.
If, at a given stage, the currently activated set of vertices is $S\subseteq\ns$,
then any unactivated $\n\in\ns\setminus S$ becomes activated in the
next stage if and only if $\infln(S)\ge\tn$.  Our goal is to understand 
how much influence different sets of vertices exert on the social network 
as a whole under this process; we can measure this by running this process 
to completion starting with a particular {\em seed set}, and seeing 
how large the final activated set is.
In some settings, we may value activating certain (sets of) vertices more 
highly, and to capture this  we define a weight
function $\wgt:\pset{\ns}\rightarrow\mathbb{R}_{+}$ on subsets of
$\ns$.  We now define the value of a seed set $S$ as
follows.  For an initially activated set $S_0$, let
$S_1^\ts,S_2^\ts,\dots,S_n^\ts$ be the activated sets after
$1,2,\dots,n=\abs{\ns}$ stages of our spreading process, when
$\ts=(\tn)_{\n\in\ns}$ is our vector of thresholds. Our goal is 
understanding the value of $\wgt(S_{n}^\ts)$ when we set $S_0=S$.
Note this depends strongly on $\ts$: the exact values of thresholds
have a significant impact on the final activated set.  If the vector 
$\ts$ can be arbitrary, finding the best seed set -- 
or even any nontrivial approximation of it -- becomes NP-Hard 
(see Section~\ref{sec:hard} for discussion and proofs of this).  
Thus, we follow the lead of Kempe et al.~\cite{kkt03} and assume that
each threshold is independently distributed as $\tn\sim\unif[0,1]$.
Then, our goal in this problem is understanding the structure of the
 function $\obj:\pset{\ns}\rightarrow\mathbb{R}_{+}$ given by
\begin{equation*}
  \obj(S) = \E_{\ts}[\;\wgt(S_{n}^{\ts})\;\vert\; S_0 = S\;],
\end{equation*}
with the goal of finding seed sets $S$ maximizing $\obj(S)$.

\paragraph*{\bf Fractional Influence Model}
A major shortcoming of the model described above is that it 
isolates the effects of influence directly applied by 
the optimizer from those of influence cascading from 
other individuals in the network.  In particular, note that every
individual in the social network is either explicitly activated by the
optimizer (and influence from their neighbors has no effect), or is
activated by influence from other individuals with no (direct) involvement
from the optimizer.  This separation is artificial, however, and in 
practical settings a clever optimizer could try to take advantage of 
known influences between the individuals they wish to affect.  For 
example, if an optimizer is already planning to activate some set 
$S$ of individuals, it should require notably less effort to ensure 
activation of any individual who is heavily influenced by the set $S$.

We propose the following modification of the previously described
influence model in order to capture this phenomenon. Rather than
selecting a set $S$ of nodes to activate, the optimizer specifies a
vector $\xs\in[0,1]^n$ indexed by V, where $\xn$ indicates the amount
of direct influence we apply to $\n$. We assume that this direct
influence is additive with influence from other vertices in the network,
in the sense that if the current activated set is $S$ in a stage of our 
process, $\n$ becomes activated in the next stage if and only if
$\infln(S)+\xn\ge\tn$.  Here, we assume that no vertices are initially
activated, that is $S_0=\emptyset$.  Note, however, that even without
contributions from other nodes, our directly-applied influence can
cause activations.  Notably, it is easy to see that
\begin{equation*}
  S_1^\ts=\{\n\in\ns:\xn\ge\tn\}.
\end{equation*}
We point out, however, that our process is not simply a matter of
selecting an initial activated set at random with marginal
probabilities $\xs$.  The influence $\xn$ we apply to $\n$ not only
has a chance to activate it at the outset, but also makes it easier
for influence from other vertices to activate it in every future stage
of the process. Lastly, we observe that this model captures the
model of Mossel and Roch~\cite{mr10} as a special case, since selecting sets to
initially activate corresponds exactly with choosing
$\xs\in\{0,1\}^n$, just with a single-round delay in the process.  This 
motivates us to term that original model as the integral
influence model, and this new model as the fractional influence model.
As before, we want to understand the structure of the expected value
of the final influenced set as a function of how we apply influence to
nodes in a graph.  We extend our function to
$\obj:[0,1]^{n}\rightarrow\mathbb{R}_{+}$ by
\begin{equation*}
  \obj(\xs)
  =
  \E_{\ts}[\;\wgt(S_n^\ts)\;\vert\;\text{we apply direct influences  $\xs$}\;]. 
\end{equation*}
We want to both understand the structure of $\obj$ and be
able to find (approximately) optimal inputs $\xs$.

\paragraph*{\bf Gap Between Integral and Fractional}  
A natural question when presented with a new model is whether it
provides any additional power over the previous one.  Here, we answer
that question in the affirmative for the fractional extension of the
Influence Maximization model.  In particular, we present two examples
here that show that fractional influence can allow a designer to
achieve strictly better results than integral influence for a
particular budget.  The first example shows that with fixed
thresholds, this gap is linear (perhaps unsurprisingly, given the
hardness of the problem under fixed thresholds).  The second example,
however, shows that even with independent uniform thresholds an
optimizer with the power to apply fractional influence can see an
improvement of up to a factor of $(1-1/e)$.
\begin{example}
  The following example shows that when thresholds are fixed, the
  optimal objective values in the fractional and integral cases can
  differ by as much as a factor of $n$, where $n$ is the number of
  vertices in the graph.  The instance we consider is a DAG consisting
  of a single, directed path of $n$ vertices.  Each edge in the path
  has weight $1/(n+1)$, and every vertex on the path has threshold
  $2/(n+1)$.  Note that since thresholds are strictly greater than
  edge weights, and every vertex, being on a simple path, has in
  degree at most one, it is impossible for a vertex to be activated
  without some direct influence being applied to it.  

  Consider our problem on the above graph with budget $1$.  In the
  integral case, we cannot split this influence, and so we may apply
  influence to -- and hence activate -- at most one vertex.  On the
  other hand, in the fractional case the following strategy guarantees
  that all vertices are activated.  Apply $2/(n+1)$ influence to the
  earliest vertex, and $1/(n+1)$ influence to the remaining $(n-1)$
  vertices.  Now, this activates the earliest vertex directly;
  furthermore, every other vertex has sufficient direct influence to
  activate it any time the vertex preceding it does.  Thus, a simple
  induction proves the claim, giving us a factor $n$ gap between  the optimal integral
  and fractional solutions.
\end{example}

\begin{example}
  Consider solving our problem on a directed graph consisting of a
  single (one-directional) cycle with $n$ vertices.  Assume that every
  edge has weight $1-K/n$, where $K$ is some parameter to be fixed later, and that thresholds on nodes are drawn from
  $\unif[0,1]$.  We consider the optimal integral and fractional
  influence to apply.

  In the fractional case, consider applying influence of exactly $K/n$
  to every node.  Note that for any node, the amount of influence we
  apply directly plus the weight on its sole incoming edge sum to $1$.
  Thus, any time a node's predecessor on the cycle becomes
  activated, the node will become activated as well.  Inductively, we
  can then see that any time at least one node is activated in the
  cycle, every node will eventually become activated.  This means that
  the expected number of activated nodes under this strategy is
  precisely
  \begin{align*}
    \ifabstract
\else
    n\cdot\Pr[\text{At least one node activates}]
    \fi
    &=n(1-\Pr[\text{Every node's threshold is above $K/n$}])\\
    &=n(1-(1-K/n)^n).
  \end{align*}
  
  In the integral case, however, we cannot spread our influence perfectly
  evenly.  Each node we activate has some chance to activate
  the nodes that are after it but before the next directly activated node in the cycle.  If we have an
  interval of length $\ell$ between directly activated nodes
  (including the initial node we activate directly as one of the $\ell$ nodes in the interval), we
  can see that the expected number of nodes activated in the interval
  is 
  \begin{align*}
    \ifabstract
\else
    \sum_{i=1}^{\ell}\Pr[\text{Node $i$ in the interval is activated}]
    \fi
    &=1+\sum_{i=2}^{\ell}\Pr[\text{Nodes $2,3,\dots,i$ have thresholds below }1-K/n]\\
    &=\sum_{i=1}^{\ell} (1-K/n)^{i-1} =\frac{1-(1-K/n)^\ell}{K/n}. 
  \end{align*}
  While this tells us the expected value for a single interval, we
  want to know the expected value summed over all intervals.
  Observing from the  above calculation that the benefit of adding another node
  to an interval is strictly decreasing in the length of the interval,
  we can see that we should always make the lengths of the intervals
  as close to equal as possible.  Noting that the lengths of the
  intervals always sum to $n$, then, we can see that the total number
  of nodes activated in expectation is bounded by
  \begin{align*}
    K\frac{1-(1-K/n)^{n/K}}{K/n}=n(1-(1-K/n)^{n/K}).
  \end{align*}
  
  Note, however, that if we choose $K\approx\ln n$, we get that 
  \begin{equation*}
    \frac{1-(1-K/n)^{n/K}}{1-(1-K/n)^n}\approx 1-1/e.
  \end{equation*}
\end{example}

\section{Reduction}
In this section, we extend the submodularity results of Mossel and Roch~\cite{mr10}
for the integral version of Influence Maximization to the fractional
version; this implies that, as in the integral version, the fractional
version admits a greedy $(1-1/e-\eps)$-approximation algorithm.  At a
high level, our approach revolves around reducing a fractional
instance to an integral one, such that evolution of the process and
objective values are preserved.  Thus, before presenting our
extension, we begin by stating the main result of \cite{mr10}.  Before
stating the theorems, however, we give definitions for the function
properties each requires.  Finally, we note that our main result of
the section (Theorem~\ref{thm:disc}) considers a {\em discretization}
of the input space; at the end of this section we show that such
discretization cannot affect our objective value by too much.

We begin by giving definitions for the following properties of set
functions.  Given a set $N$ and a function
$f:\pset{N}\rightarrow\mathbb{R}$, we say that:
\begin{itemize}
\item $f$ is {\em normalized} if $f(\emptyset)=0$;
\item $f$ is {\em monotone} if $f(S) \le f(T)$ for any $S\subseteq T \subseteq N$; and
\item $f$ is {\em submodular} if $f(S\cup\{x\})-f(S)\ge f(T\cup\{x\})-f(T)$ for any $S\subseteq T\subseteq N$ and $x\in N\setminus T$.
\end{itemize}
We say that a collection of functions satisfies the above properties if every function in the collection does.
With the above definitions in hand, we are now ready to state the following result of Mossel and Roch.
\begin{theorem}{(Restatement of \cite[Theorem 1.6]{mr10})}
  \label{thm:int}
  Let $\inst=(\ns,\infls,\wgt)$ be an instance of integral Influence Maximization.  If
  both $\wgt$ and $\infls$ are normalized, monotone, and
  submodular, then $\obj$ is as well.
\end{theorem}

We want to extend  Theorem~\ref{thm:int}  to the
fractional influence model.  We proceed by showing that
for arbitrarily fine discretizations of $[0,1]$, any instance of our
problem considered on the discretized space can be reduced to an
instance of the original problem.  Fix $N\in\mathbb{Z}_{+}$, and let
$\delta=1/N>0$ be our discretization parameter.  Let
$\Disc=\{0,\disc,2\disc,\dots,1\}$.  We consider the fractional objective function
$\obj$ restricted to the domain $\Disc^{n}$.  Lastly, let
$\unitn$ be the vector with $\disc$ in the
 component corresponding to $\n$, and $0$ in all other components.  We extend the relevant set function properties to this discretized space as follows:
\begin{itemize}
\item we say $f$ is {\em normalized} if $f(\mathbf{0})=0$;
\item we say $f$ is {\em monotone} if $\mathbf{x}\le\mathbf{y}$ implies $f(\mathbf{x})\le f(\mathbf{y})$; and
\item we say $f$ is {\em submodular} if for any $\mathbf{x}\le\mathbf{y}$, and any $\n\in\ns$, either $\yn=1$ or $f(\xs+\unitn)-f(\xs)\ge f(\ys+\unitn)-f(\ys)$,
\end{itemize}
where all comparisons and additions between vectors above are
componentwise. We get the following extension of Theorem~\ref{thm:int}.

\begin{theorem}
  \label{thm:disc}
  Let $\inst=(\ns,\infls,\wgt)$ be an instance of fractional Influence Maximization.  For any discretization $\Disc^n$ of $[0,1]^n$ (as defined
  above), 
  if both $\wgt$ and $\infls$ are normalized, monotone, and submodular,
  then $\obj$ is normalized, monotone, and submodular on $\Disc^n$.
\end{theorem}
\begin{proof}
  We prove this by reducing an instance of the (discretized) fractional problem for
  $\inst$ to an instance of the integral influence problem and then
  applying Theorem~\ref{thm:int}.  We begin by modifying $\inst$ to
  produce a new instance $\instM=(\nsM,\inflsM,\wgtM)$.  Then, we show
  that $\inflsM$ and $\wgtM$ will retain the properties of
  normalization, monotonicity, and submodularity.  Lastly, we show
  a mapping from (discretized) fractional activations for $\inst$ to
  integral activations for $\instM$ such that objective values are
  preserved, and our desired fractional set function properties for
  $\obj$ correspond exactly to their integral counterparts for the
  objective function $\objM$ for $\instM$.  The result then follows
  immediately from Theorem~\ref{thm:int}.
  
  We begin by constructing the instance $\instM$.  The key idea is
  that we can simulate fractional activation with integral activation
  by adding a set of dummy activator nodes for each original node;
  each activator node applies an incremental amount of influence on its
  associated original node.  Then, for each original node we just need
  to add the influence from activator nodes to that from the other (original)
  nodes, and truncate the sum to one.  Fortunately, both of the
  aforementioned operations preserve the desired properties.  Lastly,
  in order to avoid the activator nodes interfering with objective
  values, we simply need to give them weight zero.  With this
  overview in mind, we now define $\instM=(\nsM,\inflsM,\wgtM)$
  formally.

  First, we construct $\nsM$.  For each node $\n\in\ns$, create a set
  $\actsn=\{\actin{1},\actin{2},\dots,\actin{1/\disc}\}$ of activator
  nodes for $\n$.  We then set
  \begin{equation*}
    \nsM
    =
    \ns\cup\left({\textstyle \bigcup_{\n\in\ns}}\actsn\right).
  \end{equation*}
  We now proceed to define the functions $\inflnM$ for each $\nM\in\nsM$. If $\nM$ is an activator node for some $\n\in\ns$, we
  simply set $\inflnM\equiv0$; otherwise, $\nM\in\ns$ and we set
  \begin{equation*}
    \inflnM(S) = \min\left(\infln[\nM](S\cap\ns)+\disc\abs{S\cap\actsn[\nM]},1\right)
  \end{equation*}
  for each $S\subseteq{\nsM}$.  Lastly, we set 
  \begin{equation*}
    \wgtM(S)=\wgt(S\cap\ns)
  \end{equation*}
  for all $S\subseteq\nsM$.  Together, these make up our modified instance $\instM$.

  We now show that since $\wgt$ and $\infls$ are normalized, monotone,
  and submodular, $\wgtM$ and $\inflsM$ will be as well.  We begin with
  $\wgtM$, since it is the simpler of the two.  Now, $\wgtM$ is clearly normalized
  since $\wgtM(\emptyset)=\wgt(\emptyset)$.
  Now, fix any $S\subseteq T\subseteq\nsM$.  First, observe we have that
  $S\cap\ns\subseteq T\cap\ns$, and so
  \begin{equation*}
    \wgtM(S) = \wgt(S\cap\ns) \le \wgt(T\cap\ns) = \wgtM(T),
  \end{equation*}
  by the monotonicity of $\wgt$.
  Second, let $\nM[u]\in\nsM\setminus T$.  If $\nM[u]\in\ns$, 
  \begin{align*}
    \wgtM(S\cup\{\nM[u]\})-\wgtM(S)
    &=
    \wgt((S\cap\ns)\cup\{\nM[u]\})-\wgt(S\cap\ns) \\
    &\ge 
    \wgt((T\cap\ns)\cup\{\nM[u]\})-\wgt(T\cap\ns) \\
    &=
    \wgtM(T\cup\{\nM[u]\})-\wgtM(T),
  \end{align*}
  since $\wgt$ is submodular.  On the other hand, if
  $\nM[u]\notin\ns$, we immediately get that
  \begin{equation*}
    \wgtM(S\cup\{\nM[u]\})-\wgtM(S)
    =
    0
    =
    \wgtM(T\cup\{\nM[u]\})-\wgtM(T).
  \end{equation*}
  Thus,  $\wgtM$ is normalized, monotone, and submodular.

  Next, we show that $\inflsM$ is normalized, monotone, and
  submodular.  For $\nM\in\nsM\setminus\ns$, it follows trivially
  since $\inflsM$ is identically $0$.  In the case that $\nM\in\ns$, it is less
  immediate, and we consider each of the properties below.
  
  \begin{itemize}
  \item $\inflnM$ normalized.  This follows 
  by computing 
    \ifabstract
\else
    \begin{equation*}
      \inflnM(\emptyset)
      =\min\left(\infln[\nM](\ns\cap\emptyset)+\disc\abs{\actsn[\nM]\cap\emptyset},1\right)
      =\min\left(\infln[\nM](\emptyset)+\disc\abs{\emptyset},1\right)
      =0,
    \end{equation*}
    \fi
    since $\infln[\nM]$ is normalized.
    \item $\inflnM$ monotone.
    Let $S\subseteq T\subseteq\nsM$.  Then we
    have both $S\cap\ns\subseteq T\cap\ns$ and
    $S\cap\actsn[\nM]\subseteq T\cap\actsn[\nM]$.  Thus, we can see
    that
    \iffalse
    \begin{align*}
      \infln[\nM](\ns\cap S)
      &\le
      \infln[\nM](\ns\cap T) 
      \text{; and}\\
      \abs{\actsn[\nM]\cap S}
      &\le
      \abs{\actsn[\nM]\cap T},
    \end{align*}
    \else
    \begin{equation*}
      \infln[\nM](\ns\cap S)
      \le
      \infln[\nM](\ns\cap T) 
      \;\;\text{and}\;\;
      \abs{\actsn[\nM]\cap S}
      \le
      \abs{\actsn[\nM]\cap T},
    \end{equation*}
    \fi
    where the former follows by the monotonicity of $\infln[\nM]$.
    Combining these, we get that
    \begin{equation*}
      \infln[\nM](\ns\cap S) + \disc\abs{\actsn[\nM]\cap S}
      \le \infln[\nM](\ns\cap T) + \disc\abs{\actsn[\nM]\cap T}.
    \end{equation*}
    Note that if we replace the expression on each side of the above inequality with the minimum of $1$ and the corresponding expression, the inequality must remain valid.  Thus, we may conclude that $\inflnM(S)\le\inflnM(T)$.
    \newcommand{\diff}[1]{D_{#1}}
    \item $\inflnM$ submodular.  
    For any $S\subseteq\nsM$ and
    $\nM[u]\in\nsM\setminus S$, define the finite difference
    \begin{multline*}
      \diff{\nM[u]} \inflnM(S)=
      \big(\infln[\nM](\ns\cap(S\cup\{\nM[u]\}))+\disc\abs{(S\cup\{\nM[u]\})\cap\actsn[\nM]}\big)
      -\big(\infln[\nM](\ns\cap S)+\disc\abs{S\cap\actsn[\nM]}\big).
    \end{multline*}
    Observe that whenever $\nM[u]\notin\ns$, we immediately have that
    \begin{equation*}
      \infln[\nM](\ns\cap(S\cup\{\nM[u]\}))
      -
      \infln[\nM](\ns\cap S)
      =
      0.
    \end{equation*}
    Similarly, since $\nM[u]\notin S$, it is easy to see that the difference
    \begin{equation*}
      \abs{(S\cup\{\nM[u]\})\cap\actsn[\nM]}-\abs{S\cap\actsn[\nM]}
      =
      1
    \end{equation*}
    whenever $\nM[u]\in\actsn[\nM]$, and is $0$ otherwise.  With the
    above two observations in hand, we can simplify our finite
    difference formula as
    \begin{equation*}
      \diff{\nM[u]} \inflnM(S)
      =
      \begin{cases}
        \infln[\nM]\big(S\cup\{\nM[u]\})-\infln[\nM](S)
        &\quad\text{if $\nM[u]\in\ns$;}\\
        \disc
        &\quad\text{if $\nM[u]\in\actsn[\nM]$; and}\\
        0
        &\quad\text{otherwise.}
      \end{cases}
    \end{equation*}
    Now, fix some $S \subseteq T \subseteq \nsM$, and
    $\nM[u]\in\nsM\setminus T$.  By the submodularity of
    $\infln[\nM]$, the above equation immediately implies that
    $\diff{\nM[u]} \inflnM(S)\ge\diff{\nM[u]} \inflnM(T)$.  Applying
    case analysis similar to that for the monotonicity argument, we
    can see that this implies that 
    \begin{equation*}
      \inflnM(S\cup\{\nM[u]\})-\inflnM(S)
      \ge
      \inflnM(T\cup\{\nM[u]\})-\inflnM(T),
    \end{equation*}
    i.e.~that $\inflnM$ is submodular.

  \end{itemize}
  
  Thus,  $\inflsM$ is normalized,
  monotone, and submodular on $\nsM$, exactly as desired.  
  As such, we can apply Theorem~\ref{thm:int} to our function and get
  that for our modified instance $\instM=(\nsM,\inflsM,\wgtM)$, the
  corresponding function $\objM$ must be normalized, monotone, and
  submodular.  All that remains is to demonstrate our claimed mapping
  from (discretized) fractional activations for $\inst$ to integral
  activations for $\instM$.

  We do so as follows. For each $\n\in\ns$ and each $d\in\Disc$, let
  $\actsn^d=\{\actin{1},\actin{2},\dots,\actin{d/\disc}\}$.  Then, given the
  vector $\xs\in\Disc^n$, we set
  \begin{equation*}
    S^\xs = {\textstyle\bigcup_{\n\in\ns}}\actsn^{\xn},
  \end{equation*}
  where $\xn$ is the component of $\xs$ corresponding to the node
  $\n$.  

  We first show that under this definition we have that
  $\obj(\xs)=\objM(S^\xs)$.  In fact, as we will see the sets
  influenced will be the same not just in expectation, but for every
  set of thresholds $\ts$ for the vertices $\ns$.  Note that in the
  modified setting $\instM$ we also have thresholds for each vertex in
  $\nsM\setminus\ns$; however, since we chose $\inflnM\equiv0$ for all
  $\nM\in\nsM\setminus\ns$, and thresholds are independent draws from
  $\unif[0,1]$, we have that with probability $1$ we have
  $\inflnM(S)<\tn[\nM]$ for all $S$ and all $\nM\in\nsM\setminus\ns$.
  Thus, in the following discussion we do not bother to fix these
  thresholds, as their precise values have no effect on the spread of
  influence.

  Fix some vector $\ts$ of thresholds for the vertices in
  $\ns$. Let $S_1^\ts,\dots,S_n^\ts$ and
  $\hat{S}_1^\ts,\dots,\hat{S}_n^\ts$ be the influenced sets in each
  round in the setting $\inst$ with  influence vector $\xs$ and in the
  setting $\instM$ with influence set $S^\xs$, respectively.  We 
  \iffull can \fi show by
  induction that for all $i=0,1,\dots,n$, we have
  $\hat{S}_i^\ts\cap\ns=S_i^\ts$.  By the definition of $\wgtM$, this
  immediately implies that $\wgt(S_n^\ts)=\wgtM(\hat{S}_n^\ts)$, as desired.
  
  \ifabstract
\else
  We prove our claim by induction.  For $i=0$,
  the equality follows simply by our definitions of the processes,
  since $S_0=\emptyset$ and $\hat{S}_0=S^\xs$.  Now, assuming the
  claim holds for $i-1$, we need to show that it holds for $i$.  By
  our definition of the processes, we know that
  \begin{align*}
    S_i^\ts&=S_{i-1}^{\ts}\cup\{\n\in\ns\setminus S_{i-1}^{\ts}:\infln(S_{i-1}^\ts)+\xn\ge\tn\};
    \intertext{similarly, we have that}
    \hat{S}_i^\ts&=\hat{S}_{i-1}^{\ts}\cup\{\nM\in\nsM\setminus\hat{S}_{i-1}^{\ts}:\inflnM(\hat{S}_{i-1}^\ts)\ge\tn[\nM]\}.
  \end{align*}
  Recall, however, that for all $\nM\in\nsM\setminus\ns$, we have that
  $\inflnM\equiv0$, and it follows that
  $\hat{S}_{i}^{\ts}\setminus\ns=S^\xs$ for all $i$.  Thus, we can
  rewrite the second equality above as
  \begin{equation*}
    \hat{S}_i^\ts=\hat{S}_{i-1}^{\ts}\cup\{\n\in\ns\setminus\hat{S}_{i-1}^{\ts}:\inflnM[\n](\hat{S}_{i-1}^\ts)\ge\tn\}.
  \end{equation*}
  Consider an arbitrary $\n\in\ns\setminus
  S_{i-1}^\ts=\ns\setminus\hat{S}_{i-1}^{\ts}$.  Now, we know that
  $\n\in\hat{S}_{i}^{\ts}$ if and only if
  \begin{align*}
    \tn 
    \le \inflnM[\n](\hat{S}_{i-1}^{\ts})
    = \min(\infln(\hat{S}_{i-1}^{\ts}\cap\ns) + \disc\abs{\hat{S}_{i-1}^{\ts}\cap\actsn},1).
  \end{align*}
  Recall, however, that $\hat{S}_{i-1}^{\ts}\cap\ns=S_{i-1}^{\ts}$ by
  assumption.  Furthermore, we can compute that
  \begin{equation*}
    \abs{\hat{S}_{i-1}^{\ts}\cap\actsn}
    =\abs{S^{\xs}\cap\actsn}
    =\abs{\actsn^{\xn}}
    =\abs{\{\actin{1},\dots,\actin{\xn}\}}
    =\xn/\disc.
  \end{equation*}
  Thus, since we know that $\tn\le1$ always, we can conclude that $\n\in\hat{S}_{i}^\ts\cap\ns$ if and only if
  \begin{equation*}
    \tn \le \infln({S}_{i-1}^{\ts})+\xn,
  \end{equation*}
  which is precisely the condition for including $\n$ in
  $S_{i}^{\ts}$.  Thus, we can conclude that
  $\hat{S}_{i}^{\ts}\cap\ns=S_{i}^{\ts}$.
  \fi

  We have now shown that for all vectors of thresholds $\ts$ for
  vertices in $\ns$, with probability $1$ we have that
  $\hat{S}_{i}^{\ts}\cap\ns={S}_{i}^{\ts}$ for $i=0,1,\dots,n$.  In
  particular, note that $\hat{S}_n^{\ts}\cap\ns={S}_n^{\ts}$, and so
  $\wgtM(\hat{S}_n^{\ts})=\wgt({S}_{n}^{\ts})$.  Thus, we may conclude
  that $\objM(S^\xs)=\obj(\xs)$.  

  Lastly, we need to show that for our given mapping from
  (discretized) fractional activation vectors $\xs$ to set $S^\xs$, we
  have that the desired properties for $\obj$ are satisfied if the
  corresponding properties are satisfied for $\objM$.  So we assume
  that $\objM$ is normalized, monotone, and submodular (as, in fact,
  it must be by the above argument and Theorem~\ref{thm:int}), and show
  that $\obj$ is as well. We begin by noting that $\xs=\mathbf{0}$ implies
  $S^\xs=\emptyset$, and so $\obj(\mathbf{0})=\objM(\emptyset)=0$.  Now,
  let $\xs,\ys\in\Disc^n$ such that $\xs\le\ys$ componentwise.  First,
  we must have that $S^\xs\subseteq S^\ys$ and so
  \begin{equation*}
    \obj(\xs) = \objM(S^\xs) \le \objM(S^\ys) = \obj(\ys).
  \end{equation*}
  Second, pick some $\n\in\ns$ such that $\yn<1$.  Recall our
  definition of $\inflnM$; by inspection, we have
  $\inflnM(S)=\inflnM(T)$ any time both $S\cap\ns=T\cap\ns$ and
  $\abs{S\cap\actsn}=\abs{T\cap\actsn}$, for any $S,T\in\nsM$. Thus,
  we have $S^{\xs+\unitn}=S^\xs\cup\{\actin{(\xn/\disc)+1}\}$ and
  $S^{\ys+\unitn}=S^\ys\cup\{\actin{(\yn/\disc)+1}\}$.  So we have
  \begin{align*}
    \obj(\xs+\unitn)-\obj(\xs)
    &=\objM(S^\xs\cup\{\actin{(\xn/\disc)+1}\}) - \objM(S^\xs)\\
    &=\objM(S^\xs\cup\{\actin{(\yn/\disc)+1}\})-\objM(S^\xs)\\
    &\ge \objM(S^\ys\cup\{\actin{(\yn/\disc)+1}\}) - \objM(S^\ys)\\
    &=\obj(\ys+\unitn)-\obj(\ys).
  \end{align*}
  Thus,  $\obj$ has the claimed properties on $\Disc^n$; the result follows. 
\end{proof}

\newcommand{\epsn}[1][\n]{{\boldsymbol\eps}_{#1}}
In fact, we can use the same technique to achieve the following
extension to fully continuous versions of our properties.  We define
the following properties for $\obj$ on the continuous domain
$[0,1]^n$:
\begin{itemize}
\item we say $f$ is {\em normalized} if $f(\mathbf{0})=0$;
\item we say $f$ is {\em monotone} if $\mathbf{x}\le\mathbf{y}$ implies $f(\mathbf{x})\le f(\mathbf{y})$; and
\item we say $f$ is {\em submodular} if for any
  $\mathbf{x}\le\mathbf{y}$, any $\n\in\ns$, and for any $\eps>0$ such
  that $\yn+\eps\le1$, we have that
  $f(\xs+\epsn)-f(\xs)\ge f(\ys+\epsn)-f(\ys)$,
\end{itemize}
where $\epsn$ is the vector with a value of $\eps$ in the coordinate
corresponding to $\n$ and a value of $0$ in all other coordinates.  As
before, all comparisons and additions between vectors above are
componentwise.  
\ifabstract
\else   
The same techniques immediately give us the following theorem.
\fi

\begin{theorem}
  \label{thm:cont}
  Let $\inst=(\ns,\infls,\wgt)$ be an instance of our problem.  If
  both $\wgt$ and $\infls$ are normalized, monotone, and submodular,
  then $\obj$ is normalized, monotone, and submodular on $[0,1]^n$.
\end{theorem}
\iffull
\begin{proof}
  We use the exact same technique as in the proof of Theorem~\ref{thm:disc}.
  The only difference is how we define the activator nodes in our
  modified instance $\instM$.  Here, rather than trying to model the
  entire domain, we simply focus on the points we want to verify our
  properties on.  To that end, fix some $\xs\le \ys$, as well as an
  $\eps>0$ and some $\n\in\ns$.  We will only define three activator
  nodes here: $a^\xs$, $a^{\xs-\ys}$, and $a^\eps$.  The first two
  contributes amounts of $\xn[\n']$ and $(\yn[\n']-\xn[\n'])\ge0$,
  respectively, to the modified influence function for vertex
  $\n'\in\ns$.  The last contributes an amount of $\eps$ to the
  influence function for vertex $\n$, and makes no contribution to any
  other influence functions.  As before, all influence functions get
  capped at one.  Our modified weight function is defined exactly as
  before, and it is easy to see that exactly the same argument will
  imply that the modified weight and influence functions will be
  normalized, monotone, and submodular, allowing us to apply
  Theorem~\ref{thm:int}.

  Thus, all that remains is to relate the function values between the
  original and modified instances.  Note, however, that here it is
  even simpler than in the discretized case.  If $\objM$ is the
  objective for $\instM$, then we can compute that:
  \begin{align*}
    \obj(0)&=\objM(\emptyset)\text{;}\\
    \obj(\xs)&=\objM(\{a^\xs\})\text{;}\\
    \obj(\ys)&=\objM(\{a^\xs,a^{\ys-\xs}\})\text{;}\\
    \obj(\xs+\epsn)&=\objM(\{a^\xs,a^\eps\})\text{; and}\\
    \obj(\ys+\epsn)&=\objM(\{a^\xs,a^{\ys-\xs},a^\eps\})\text{.}
  \end{align*}
  We can use the above inequalities to show the desired qualities for
  $\obj$, via a simple case analysis and their discrete counterparts for
  $\objM$.
\end{proof}
\fi

One concern with discretizing the space we optimize over as in
Theorem~\ref{thm:disc} is what effect the discretization has on the
objective values that can be achieved.  
\iffull
As the following theorem shows, however, we can only lose a
$\disc n/K$ factor from our objective when we discretize the space to
multiples of $\disc$.  
\else
As the following theorem shows, however, that we can only lose a
$\disc n/K$ factor from our objective when we discretize the space to
multiples of $\disc$; we defer the proof to the full version of the paper.
\fi
\begin{theorem}
  \label{thm:approx}
  Let $\inst=(\ns,\infls,\wgt)$ be an instance of our problem.  Then
  for any discretization $\Disc^n$ of $[0,1]^n$ (as defined above), if
  $\obj$ is normalized, monotone, and submodular on $\Disc^n$, we have
  that
  \begin{equation*}
    \max_{\substack{\xs\in\Disc^n:\\\norm[1]{\xs}\le K }}\obj(\xs)
    \ge
    (1-\disc{\textstyle\frac{n}{K}})\max_{\substack{\xs\in{[0,1]}^n:\\\norm[1]{\xs}\le K}}\obj(\xs),
\end{equation*}
  for any $K$.
\end{theorem}
\iffull
\begin{proof}
  Let $\xs^\ast$ be an optimal solution to our problem on $[0,1]^n$, i.e. we have
  \begin{equation*}
    \xs^\ast=\argmax_{\substack{\xs\in{[0,1]}^n:\norm[1]{\xs}\le K}}\obj(\xs).
  \end{equation*}
  Let $\bar{\xs}^\ast$ be the result of rounding $\xs^\ast$ up
  componentwise to the nearest element of $\Disc^n$.  Formally, we
  define $\bar{\xs}^{\ast}$ by
  $\bar{x}^{\ast}_{\n}=\min\{d\in\Disc:d\ge{x}^\ast_v\}$.  Note that
  by monotonicity, we must have that
  $\obj(\bar{\xs}^\ast)\ge\obj(\xs^\ast)$; we also have that
  $\norm[1]{\bar{\xs}^\ast}\le\norm[1]{\xs^\ast}+\disc n$.  Now,
  consider constructing $\bar{\xs}^\ast$ greedily by adding $\disc$ to
  a single coordinate in each step.  Formally, set $\xs^{0}={\mathbf0}$,
  and for each $i=1,2,\dots,\norm[1]{\bar{\xs}^\ast}/\disc$ set
  \begin{equation*}
    \xs^{i}=\xs^{i-1}+\unitn\text{ for some }\n\in\argmax_{\n:\;x^{i-1}_{\n}<\bar{x}^\ast_\n}(\obj(\xs^{i-1}+\unitn)-\obj(\xs^{i-1})),
  \end{equation*}
  where (as before) $\unitn$ is a vector with $\disc$ in the component
  corresponding to $\n$ and $0$ in all other components.  Note that
  the submodularity of $\obj$ implies that
  $\obj(\xs^{i})-\obj(\xs^{i-1})$ is decreasing in $i$.  An immediate
  consequence of this is that, for any $i$, we have that
  \begin{equation*}
    \obj(\xs^{i}) \ge
    \frac{i}{\norm[1]{\bar{\xs}^{\ast}}/\disc}\obj(\bar{\xs}^{\ast}).
  \end{equation*}
  Invoking the above for $i=K/\disc$ we get that
  \begin{equation*}
  \obj(\xs^{K/\disc}) 
  \ge \frac{K/\disc}{\norm[1]{\bar{\xs}^{\ast}}/\disc}\obj(\bar{\xs}^{\ast})
  \ge \frac{K}{K\!+\!\disc n}\obj(\bar{\xs}^{\ast})
  \ge (1-\disc{\textstyle\frac{n}{K}})\obj(\bar{\xs}^{\ast}).
  \end{equation*}
  We observe that $\norm[1]{\xs^{K/\disc}}=K$, and
  $\xs^{K/\disc}\in\Disc^n$, and so the desired theorem follows.
\end{proof}
\fi

By combining the above theorems, we can apply the same results
as~\cite{kkt03} to get the following corollary (see that paper and
references therein for details).
\begin{corollary}
  There exists a greedy $(1-1/e-\eps)$-approximation for maximizing
  $\obj(\xs)$ in both the discretized and continuous versions of the
  fractional influence model.
\end{corollary}

\section{DAGs}
\label{sec:dag}
In this section, we focus on a special case of fractional influence model called the linear influence model, and argue
that some aspects of the problem become simpler on DAGs when the thresholds are uniformly distributed.  
Our interest in DAGs is motivated by the fact that the hardness results presented in the next section all hold for DAGs.  
In the linear variant of the problem, our influence functions are computed as
follows.  We are given a digraph $G=(\ns,E)$ and a weight function $\ws$
on edges.  We use $\nbdin(v)$ and $\nbdout(v)$ to denote the sets of
nodes with edges to $v$ and edges from $v$, respectively.  Then, we denote
the influence function $\infln$ for $v$ by
\begin{equation*}
  \infln(S)=\sum_{u\in S \cap \nbdin(\n)}\w{uv}.
\end{equation*}
In this model, we assume that $\sum_{u\in\nbdin(v)}\w{uv}\le1$ always.  
Similar to the fractional influence model, our goal is to pick an influence vector
$\xs\in[0,1]^{\abs{V}}$ indexed by $V$ to maximize
\begin{equation*}
  \obj(\xs)
  =
  \E_{\ts}[\;\abs{S_n^\ts}\;\vert\;\text{we apply direct influences  $\xs$}\;],
\end{equation*}
where $S_1^\ts,\dots,S_n^\ts$ is the sequence of sets of nodes
activated under thresholds $\ts$ and direct influence $\xs$.  We
sometimes abuse notation and use $\obj(S)$ to denote $\obj$ applied to
the characteristic vector of the set $S\in2^\ns$.
Given a DAG $G=(V,E)$ and a fractional influence vector $\xs\in[0,1]^{\abs{V}}$ indexed by $V$, we define the sets
\begin{align*}
  \Infl(\xs)&=\{v\in V: \xn > 0\}\text{, and}\\
  \Sat(\xs)&=\{{\textstyle v\in V: \xn+\sum_{u\in\nbdin(v)}\w{uv} > 1}\},
\end{align*}
as the sets of nodes {\em influenced} by $\xs$ and {\em
  (over-)saturated} by $\xs$.  Note that
$\Sat(\xs)\subseteq\Infl(\xs)$.  Next, we show that under specific circumstances, $\obj$ becomes a linear function and therefore the influence maximization problem efficiently solvable. 

\begin{theorem}
  Given a DAG $G$ and influence vector $\xs$, if G contains no path
  from an element of $\Infl(\xs)$ to any element of $\Sat(\xs)$, then
  we have that
\begin{equation*}
\textstyle  \obj(\xs) = \sum_{v\in\ns}\xn\obj({\mathbf 1}_v),
\end{equation*}
and therefore the influence maximization problem can be solved efficiently restricted to vectors $\xs$ meeting this condition.
\end{theorem}
\begin{proof}
We prove this by induction on the number of vertices.  In the case
that $\ns$ contains only a single vertex, the claim is trivial.
Otherwise, let $G=(V,E)$ and $\xs$ satisfy our assumptions, with
$\abs{V}=n>1$, and assume out claim holds for any DAG with $(n-1)$ or
fewer nodes.  Let $s\in\ns$ be a source vertex (i.e. have in-degree
$0$) in $G$.  Now, if $s\notin\Infl(\xs)$, we know that $s$ is never
activated.  Let $\hat{\obj}$ and $\hat{\xs}$ be $\obj$ on $G$
restricted to $\ns\setminus s$ and $\xs$ restricted to $\ns \setminus
s$, respectively, and observe that we may apply our induction
hypothesis to $\hat{\obj}(\hat{\xs})$ since removing $s$ from $G$
cannot cause any of the requirements for our theorem to become
violated.  Thus, since $\xn[s]=0$, we can see that
\begin{equation*}
  \obj(\xs) 
  = \hat{\obj}(\hat{\xs})
  =\sum_{v \in \ns \setminus s}\xn\hat{\obj}({\mathbf 1}_v)
  =\sum_{v \in \ns}\xn\obj({\mathbf1}_v).
\end{equation*}

Now, assume that $s\in\Infl(\xs)$.  Recall that our conditions on
$G$ ensures it contains no path from $s$ to any
elements of $\Sat(\xs)$.  Furthermore, this implies 
none of the nodes in ${\nbdout(s)}$ have paths to elements of
$\Sat(\xs)$ either, and so applying influence to them does not
violate the assumptions of our inductive hypothesis, as long as
we ensure we do not apply enough influence to (over-)saturate them. 

\newcommand{\hatxn}[1][\n]{\hat{x}_{#1}}
\newcommand{\hatyn}[1][\n]{\hat{y}_{#1}}

In order to prove our claim, we focus on $G$ restricted to
$V\setminus\{s\}$, call it $\hat{G}$.  Let $\hat{\obj}$ be $\obj$ over
$\hat{G}$, and consider the following two influence vectors for
$\hat{G}$.  Define $\hat{\xs}$ to simply be the restriction of $\xs$
to $\hat{G}$; define $\hat{\ys}$ by $\hatyn=\w{sv}$ if $v\in\nbdout(s)$
and $0$ otherwise.  Letting $\hat{\Infl}$ and $\hat{\Sat}$ be $\Infl$
and $\Sat$, respectively, restricted to $\hat{G}$, we have
\begin{equation}
  \label{eqn:satCompare}
  \left.
  \begin{aligned}
    \hat{\Infl}(\hat{\xs}),\hat{\Infl}(\hat{\ys}),\hat{\Infl}(\hat{\xs}+\hat{\ys})
    &\subseteq \Infl(\xs)\cup\nbdout(s)\text{, and}\\
    \hat{\Sat}(\hat{\xs}),\hat{\Sat}(\hat{\ys}),\hat{\Sat}(\hat{\xs}+\hat{\ys})
    &\subseteq \Sat(\xs).
  \end{aligned}
  \quad\right\}
\end{equation}
The observation that gives the above is that, compared to $\xs$, the
only vertices with increased influence applied to them are the
elements of $\nbdout(s)$, and the amounts of these increases are
precisely balanced by the removal of $s$ (and its outgoing edges) from
$\hat{G}$.  In particular, note that for any $v\in\ns\setminus\{s\}$,
by our definition of $\hat{\ys}$ we have that
\begin{equation*}
  \xn+\sum_{u\in\nbdin(v)}\w{uv} 
  = \hatxn+\hatyn +  \sum_{u\in\nbdin(v)\setminus\{s\}}\w{uv}.
\end{equation*}

As previously noted $G$ contains no paths from an element of
$\nbdout(s)$ to any element of $\Sat(\xs)$; this combined
with~\eqref{eqn:satCompare} allows us to conclude that we may apply
our induction hypothesis to $\hat{G}$ with any of $\hat{\xs}$,
$\hat{\ys}$, or $\hat{\xs}+\hat{\ys}$.  We proceed by showing that for
any vector $\ts$ of thresholds for $G$ (and its restriction to
$\hat{G}$), we have that the set activated under $\xs$ in $G$ always
corresponds closely to one of the sets activated by $\hat{\xs}$ or
$(\hat{\xs}+\hat{\ys})$ in $\hat{G}$.  To that end, fix any vector
$\ts$.  We consider the cases where $\xn[s]\ge\tn[s]$ and
$\xn[s]<\tn[s]$ separately.

We begin with the case where $\xn[s]<\tn[s]$, since it is the simpler
of the two.  Let $S_0^\ts,\dots,S_n^\ts$ and
$\hat{S}_0^\ts,\dots,\hat{S}_n^\ts$ denote the sets activated in $G$
under $\xs$ and in $\hat{G}$ under $\hat{\xs}$, respectively, in
stages $0,\dots,n$.  Note that since $s$ is a source, and
$\xn[s]<\tn[s]$, we know that $s\notin S_i^\ts$ for all $i$.  However,
this means that every node in $\ns\setminus\{s\}$ has both the same
direct influence applied to it under $\xs$ and $\hat{\xs}$, and the
same amount of influence applied by any activated set in both $G$ and
$\hat{G}$.  So we can immediately see that since
$S_0^\ts=\emptyset=\hat{S}_0^\ts$, by induction we will have that
$S_i^\ts=\hat{S}_i^\ts$ for all $i$, and in particular for $i=n$.

The case where $\xn[s]\ge\tn[s]$ requires more care.  Let
$S_0^\ts,\dots,S_n^\ts$ and $\hat{S}_0^\ts,\dots,\hat{S}_n^\ts$ denote
the sets activated in $G$ under $\xs$ and in $\hat{G}$ under
$\hat{\xs}+\hat{\ys}$, respectively, in stages $0,\dots,n$.  Note that
our assumption implies that $s$ will be activated by our direct
influence in the first round, and so we have $s\in\hat{S}_i^\ts$ for
all $i\ge1$.  Fix some $v\in\ns$, $v\neq s$, and let $f_v(S)$ and
$\hat{f}_v(S)$ denote the total influence -- both direct and cascading
-- applied in $G$ and $\hat{G}$, respectively, when the current
active set is $S$.  Then, for any
$S\subseteq\ns\setminus\{s\}$ we have 
\begin{equation}
  \label{eqn:2}
  \begin{aligned}
    \hat{f}_v(S)
    &=\hatxn+\hatyn+{\sum_{\substack{u\in\nbdin(v)\\ u\in S}}}\w{uv}\\
    &=\xn+{\sum_{\substack{u\in\nbdin(v)\\ u\in S\cup\{s\}}}}\w{uv}
    =f_v(S\cup\{s\}).
  \end{aligned}
\end{equation}
Furthermore, note that both $f_v$ and $\hat{f}_v$ are always monotone
nondecreasing.  While we cannot show that $S_i^\ts=\hat{S}_i^\ts$ for
all $i$ in this case, we will instead show that
$S_i^\ts\setminus\{s\}\subseteq\hat{S}_i^\ts\subseteq
S_{i+1}^\ts\setminus\{s\}$ for all $i=0,\dots,n-1$.  Recall that the
propagation of influence converges by $n$ steps.  That is, if we
continued the process for an additional step to produce activated sets
$S_{n+1}^\ts$ and $\hat{S}_{n+1}^\ts$, we would have that
$S_{n+1}^\ts=S_n^\ts$ and $\hat{S}_{n+1}^\ts=\hat{S}_{n}^\ts$.
However, our claim would extend to this extra stage as well, and so we
conclude that we must have that $S_n^\ts=\hat{S}_n^\ts\cup\{s\}$.  We
prove our claim inductively.  First, observe that it holds trivially
for $i=0$, since we have $S_0^\ts=\hat{S}_0^\ts=\emptyset$, and
previously observed that $s\in S_1^\ts$. Now, the claim holds for some
$i$.  Note, however, that by \eqref{eqn:2} and monotonicity we must
have that for all $v\in\ns$, $v\neq s$
\begin{align*}
  f_v(S_i^\ts)
  &=\hat{f}_v(S_i^\ts\setminus\{s\})
  \le\hat{f}_v(\hat{S}_i^\ts) \\
  &\le\hat{f}_v(S_{i+1}^\ts\setminus\{s\})
  =f_v(S_{i+1}^\ts).
\end{align*}
From the above, we can conclude that
$S_{i+1}\setminus\{s\}\subseteq\hat{S}_{i+1}^\ts\subseteq
S_{i+2}^\ts\setminus\{s\}$ since such a $v$ in included in each of the
above sets if and only if $f_v(S_i^\ts)$, $\hat{f}_v(\hat{S}_i^\ts)$,
or $f_v(S_{i+1}^\ts)$, respectively, exceeds $\tn$.

Thus, by observing that $\tn[s]$ is an independent draw from
$\unif[0,1]$, we can see that taking expectations over $\ts$ and
conditioning on which of $\tn[s]$ and $\xn[s]$ is larger gives us that
\begin{align*}
  \obj(\xs) 
  &= (1-\xn[s])\hat{\obj}(\hat{\xs}) + \xn[s](1+\hat{\obj}(\hat{\xs}+\hat{\ys}))\\
  &= \sum_{\substack{v\in\ns\\v\neq s}}\xn\obj({\mathbf1}_v)+\xn[s](1+\hat{\obj}(\hat{\ys})).
\end{align*}
We complete our proof by observing that $\obj({\mathbf1}_s)$ is
precisely equal to $1+\hat{\obj}(\hat{\ys})$.  We can show this, once again,
by coupling the activated sets under any vector $\ts$ of
thresholds.  In particular, let $S_0^\ts,\dots,S_n^\ts$ and
$\hat{S}_0^\ts,\dots,\hat{S}_n^\ts$ denote the sets activated in $G$
under ${\mathbf 1}_s$ and in $\hat{G}$ under $\hat{\ys}$, respectively, in
stages $0,\dots,n$.  Arguments identical to those made
above allow us to conclude that for all $i$, we have 
$S_{i+1}^\ts=\hat{S}_i^\ts\cup\{s\}$.  Thus, by again noting that
influence cascades converge after $n$ steps we have
$S_{n}^\ts=\hat{S}_n^\ts\cup\{s\}$, and taking expectations with
respect to $\ts$ gives precisely the desired equality.

Since we have that $\obj$ is linear, we can maximize $\obj(\xs)$ by 
greedily applying influence to the vertex with the highest $\obj({\mathbf 1}_v)$ until the budget is exhausted. Estimating $\obj({\mathbf 1}_v)$ can be done by repeating the process with influence vector ${\mathbf 1}_v$ several times, and averaging the number of activated nodes in these trials.   
\end{proof}

In particular, note that the above theorem says that if we want to
find the best set of vertices to influence among those in the first
layer of a multi-layer DAG, we can efficiently solve this exactly.\footnote{Note that the hardness results Theorem~\ref{thm:fixed_nphard} and
Corollary~\ref{cor:apphard} in the next section hold exactly for such DAG problems, but with fixed thresholds instead of uniform ones.}

\iffull
We may also express our optimization problem on DAGs in the integral
case as the following MIP:
\begin{alignat*}{6}
  \mathclap{\text{maximize} \sum_v (X_v+Y_v) \text{ subject to}}\\
  X_v+Y_v &\le1 &\qquad&\forall v\\
  Y_v - \sum_{u\in\nbdin(v)}\w{uv}(Y_u+X_u)&\le0&&\forall v\\
  \sum_v X_v &\le K\\
  X_v&\in\{0,1\} &&\forall v\\
  Y_v &\in [0,1] && \forall v
\end{alignat*}
\fi


\section{Hardness}
\label{sec:hard}
\newcommand{\layer}[1]{L_{#1}}
In this section, we present NP-hardness and inapproximability results in the linear influence model. 
We assume that thresholds are not chosen from a distribution, and they are fixed and given as part of the input. 
We note that this is the main assumption that makes our problem intractable, and to achieve reasonable algorithms, 
one has to make some stochastic (distributional) assumptions on the thresholds. 
In Section~\ref{sec:dag}, we introduced the linear influence model as a special case of the fractional influence model, but it
makes sense to define it as a special case of the integral influence model as well. 
In the fractional linear influence model, we are allowed to apply any influence vector  $\xs \in [0,1]^n$ on nodes. 
By restricting the influence vector $\xs$ to be in $\{0,1\}^n$ (a binary vector), 
we achieve the integral version of the linear influence model. 
Our hardness results in Theorem~\ref{thm:fixed_nphard} and Corollary~\ref{cor:apphard} work for both fractional and integral versions of the linear influence model. 
We start by proving that the linear influence model is NP-hard with a reduction from Independent Set  in Theorem~\ref{thm:fixed_nphard}. 
We strengthen this hardness result in Corollary~\ref{cor:apphard} by showing that an $n^{1-\epsilon}$ approximation algorithm for the linear influence problem 
yields an exact algorithm for it as well for any constant $\epsilon > 0$, and therefore even  an $n^{1-\epsilon}$ approximation algorithm is NP-hard to achieve. 
At the end, we show that it is NP-hard to achieve any approximation factor better than $1-1/e$  in the Triggering model (a generalization of the linear threshold model introduced in \cite{kkt03}). 
We will elaborate on the Triggering Model and this hardness result at the end of this section. 

\begin{theorem}
  \label{thm:fixed_nphard}
  If we allow arbitrary, fixed thresholds, it is NP-hard to compute
  for a given instance of the integral linear influence problem $(G,k,T)$ (graph $G$, budget $k$, and a target goal $T$)
  whether or not there exists a set $S$ of $k$ vertices in $G$
  such that $\sigma(S)\ge T$.  Furthermore, the same holds in the
  fractional version of the problem (instead of a set $S$ of size $k$, we should look for a influence vector with $\ell_1$ norm equal to $k$ in the fractional case).  Additionally, this holds even when $G$ is a two-layer DAG and only vertices in the first layer may be influenced.
\end{theorem}
\iffull
\begin{proof}
  We show hardness by reducing from Independent Set.  Given a problem
  instance $(G,k)$ of IS, we construct a two-layer DAG as follows.
  Let $G=(V,E)$ denote the vertices and (undirected) edges of $G$.
  The first layer $\layer{1}$ consists of one vertex for every vertex
  $v\in{V}$; we abuse notation and refer to the vertex in $\layer{1}$
  corresponding to $v\in{V}$ as $v$ as well.  The second layer
  contains vertices based on the edges in $E$.  For each unordered
  pair of vertices $\{u,v\}$ in $V$, we add vertices to the second
  layer $\layer{2}$ based on whether $\{u,v\}$ is an edge in $G$: if
  $\{u,v\}\in{E}$, then we add a single vertex to $\layer{2}$ with
  (directed) edges from each of $u,v\in\layer{1}$ to it; if
  $\{u,v\}\notin{E}$, then we add two vertices to $\layer{2}$, and add
  (directed) edges going from $u\in\layer{1}$ to the first of these
  and from $v\in\layer{1}$ to the second of these.  We set all
  activation thresholds and all edge weights in our new DAG to $1/2$.
  We claim that there exists a set $S\subseteq\layer{1}\cup\layer{2}$
  satisfying $\abs{S}\le{k}$ and $\sigma(S)\ge{kn}$ if and only if $G$
  has an independent set of size $k$.

  First, we note that in our constructed DAG, sets
  $S\subseteq\layer{1}$ always dominate sets containing elements
  outside of $\layer{1}$, in the sense that for any
  $T\subseteq\layer{1}\cup\layer{2}$ there always exists a set
  $S\subseteq\layer{1}$ such that $\abs{S}\le\abs{T}$ and
  $\sigma(S)\ge\sigma(T)$.  Consider an arbitrary such $T$.  Now,
  consider any vertex $v\in T\cap\layer{2}$.  By construction, there
  exists some $u\in\layer{1}$ such that $(u,v)$ is an edge in our DAG.
  Note that $\abs{T\setminus\{v\}\cup\{u\}}\le\abs{T}$ and
  $\sigma(T\setminus\{v\}\cup\{u\})\ge\sigma(T)$.  Thus, if we
  repeatedly replace $T$ with $T\setminus\{v\}\cup\{u\}$ for each such
  $v$, we eventually with have the desired set $S$.

  With the above observation in hand, we can be assured that there
  exists a set $S$ of $k$ vertices in our constructed DAG such that
  $\sigma(S)\ge nk$ if and only if there exists such an
  $S\subseteq\layer{1}$.  Recall how we constructed the second layer
  of our DAG: each vertex $v\in\layer{1}$ has precisely $(n-1)$
  neighbors; and two vertices $u,v\in\layer{1}$ share a neighbor if
  and only if they are neighbors in the original graph $G$, in which
  case they have exactly one shared neighbor.  Thus, we can see that
  for any set of vertices $S\subseteq\layer{1}$, we have that
  \begin{equation*}
    \sigma(S)
    =
    n\abs{S}-\abs{\{\{u,v\}\in E: u,v\in S\}}.
  \end{equation*}
  Thus, we can see that for any set $S\subseteq\layer{1}$, we have
  $\sigma(S)\ge n\abs{S}$ if and only if $\{u,v\}\notin E$ for any
  $u,v\in S$, i.e.~$S$ is an independent set in $G$.  The main
  claim follows.

  Furthermore, recall that in our constructed DAG, every edge weight
  and threshold was exactly equal to $1/2$.  It is not hard to see,
  therefore, that in the fractional case it is never optimal to place
  an amount of influence on a vertex other than $0$ or $1/2$.  It
  follows, therefore, that there is a $1-1$ correspondence between
  optimal optimal solutions in the integral case with budget $k$ and
  in the fractional case with budget $k/2$.  Thus, as claimed, the
  hardness extends to the fractional case.
  \end{proof}
\fi

\begin{corollary}\label{cor:apphard}
  If we allow arbitrary, fixed thresholds, it is NP-hard to
  approximate the linear influence problem to within a factor of
  $n^{1-\eps}$ for any $\eps>0$.  Furthermore, the same holds for the
  fractional version of our problem.
Additionally, this holds even when $G$ is a three-layer DAG and only vertices in the first layer may be influenced.
\end{corollary}
\iffull
\begin{proof}
  We show that given an instance $(G,k)$ of Target Set Selection, and
  a target $T$, we can construct a new instance $(G',k)$, such that if
  we can approximate the optimal solution for the new instance
  $(G',k)$ to within a factor of $n^{1-\eps}$, then we can tell
  whether the original instance $(G,k)$ had a solution with objective
  value at least $T$.  The claim then follows by applying
  Theorem~\ref{thm:fixed_nphard}.  
  
  Fix some $\delta>0$.  Let $n$ be the number of vertices in $G$.  Note
  that we must have that $0<k<T\le n$, since if any one of these
  inequalities fails to hold, the question of whether or not $(G,k)$
  has a solution with objective value at least $T$ can be answered
  trivially.  Let $N=\lceil{(2n^2)^{1/\delta}}\rceil$; we construct $G'$
  from $G$ by adding $N$ identical new vertices to it.  Let $v$ be one
  of our new vertices.  For every vertex $u$ that was present in $G$,
  we add an edge from $u$ to $v$ in $G'$, with weight $1/n$.  We set
  the threshold of $v$ to be precisely $T/n$.

  Consider what the optimal objective value in $(G,k)$ implies about
  the optimal objective value in $(G',k)$.  If there exists some
  solution to the former providing objective value at least $T$, then
  we can see that the same solution will activate every one of the new
  vertices in $G'$ as well, and so produce an objective value of at
  least $T+N$.  On the other hand, assume every solution to $G$
  has objective value strictly less than $T$.  Note that in the case
  of fixed thresholds, the activation process is deterministic, and so
  we may conclude that every solution has objective value at most
  $T-1$.  Now, this means that no matter what choices we make in $G'$
  about the vertices inherited from $G$, every one of the new vertices
  will require at least $1/n$ additional influence to become
  activated.  Thus, no solution for $(G',k)$ can achieve objective
  value greater than $(T-1)+kn$, in either the integral or
  fractional case.  By our choice of $N$, however, we can then
  conclude that the optimal solution for $(G',k)$ in the first case has value at least 
  \begin{equation*}
    T+N > N \ge (2n^2)^{1/\delta} > (T-1+kn)^{1/\delta},
  \end{equation*}
  the value of the optimal solution for $(G',k)$ in the latter case
  raised to the power of $1/\delta$.  Thus, for any fixed $\eps>0$, we
  can choose an appropriate $\delta>0$ such the new instance $(G',k)$
  has increased in size only polynomially from $(G,k)$, but applying
  an $n^{1-\eps}$ approximation to $(G',k)$ will allow us to
  distinguish whether or not $(G,k)$ had a solution with objective
  value at least $T$, exactly as desired.
\end{proof}
\fi

Before stating Theorem~\ref{thm:trigger}, we  define the triggering model introduced in \cite{kkt03}. 
In this model, each node $v$ independently chooses
a random triggering set $T_v$ according to some distribution
over subsets of its neighbors. To start the process, we target
a set A for initial activation. After this initial iteration, an
inactive node v becomes active in step $t$ if it has a neighor
in its chosen triggering set $T_v$ that is active at time $t-1$.
For our purposes, the distributions of triggering sets have support size one (deterministic triggering sets).
We also show that our hardness result even holds when the size of these sets is two.  

\begin{theorem}\label{thm:trigger}
  It is NP-hard to approximate the linear influence problem to within any
  factor better than $1-1/e$, even in the Triggering model where
  triggering sets have size at most $2$.  Furthermore, this holds even when $G$ is a two-layer DAG and only nodes in the first layer may be influenced.
\end{theorem}
\iffull
\begin{proof}
  We prove this by reducing from the Max Coverage problem, which is
  NP-hard to approximate within any factor better than $1-1/e$.  Let
  $(\mathcal{S},k)$ be an instance of Max Coverage, where
  $\mathcal{S}=\{S_1,\dots,S_m\}$ and $S_j\subseteq[n]$ for each
  $j=1,\dots,m$.  We begin by showing a reduction to an instance of
  Target Set Selection in the Triggering model; later, we argue that
  we can do so while ensuring that triggering sets have size at most
  $2$.  

  We construct a two layer DAG instance of Target Set Selection as
  follows.  First, fix a large integer $N$; we will pick the exact
  value of $N$ later. The first layer $\layer{1}$ will contain $m$
  vertices, each corresponding to one of the sets in $\mathcal{S}$.
  The second layer $\layer{2}$ contains $nN$ vertices, $N$ of which
  correspond to each $i\in[n]$.  We add directed edges from the vertex
  in $\layer{1}$ corresponding to $S_j$ to all $N$ vertices in
  $\layer{2}$ corresponding to $i$ for each $i\in S_j$.  We set all
  thresholds and weights in the DAG to $1$.  Note that this
  corresponds exactly to the triggering model, where each vertex in
  the first layer has an empty triggering set and each vertex in the
  second layer has a triggering set consisting of exactly the nodes in
  $\layer{1}$ corresponding to $S_j$ that contain it.  This completes
  the description of the reduction.

  Now, we consider the maximal influence we can achieve by selecting
  $k$ vertices in our constructed DAG.  First, we note that we may
  assume without loss of generality that we only consider choosing
  vertices from $\layer{1}$.  This is because we can only improve the
  number of activated sets by replacing any vertex from $\layer{2}$
  with a vertex that has an edge to it; if we have already selected
  all such vertices, then we can simply replacing it with an arbitrary
  vertex from $\layer{1}$ and be no worse off.  Note, however, that if
  we select some set $W\subseteq\layer{1}$ of vertices to activate, we
  will have that
  \begin{equation*}
    \sigma(W)=\abs{W}+N\abs{\cup_{j\in W}S_j}.
  \end{equation*}
  Let $W^\ast\in\argmax_{W\subset\layer{1}:\abs{W}\le k}\sigma(W)$.
  Now, if we have an $\alpha$-approximation algorithm for Target Set
  Selection, we can find some $W\subseteq\layer{1}$ such that
  $\abs{W}\le k$ and $\sigma(W)\ge\alpha\sigma(W^\ast)$.  But this
  means that
  \begin{align}
    \notag
    \abs{W}+N\abs{\cup_{j\in W}S_j} 
    &\ge 
    \alpha\left(\abs{W^\ast}+N\abs{\cup_{j\in W^\ast}S_j}\right),
    \intertext{which implies that}
    \abs{\cup_{j\in W}S_j} 
    &\ge 
    \label{eqn:approx_loss}
    \alpha\abs{\cup_{j\in W^\ast}S_j} - m/N.
  \end{align}
  Thus, for any $\eps>0$, by picking $N=\lceil{m/\eps}\rceil$ we can
  use our $\alpha$-approximation algorithm for Target Set Selection to
  produce an $\alpha$-approximation for Max Coverage with an additive
  loss of $\eps$.  Since the objective value for our problem is
  integral, we may therefore conclude it is NP-hard to approximate
  Target Set Selection within a factor of $1-1/e$.

  In the above reduction, our targeting sets could be as large as $m$.
  We know show that we can, in fact, ensure that no targeting set has
  size greater than $2$.  In particular, the key insight is that
  activation effectively functions as an OR-gate over the targeting
  set.  We can easily replace an OR-gate with fan-in of $f$ by a tree
  of at most $\log(f)$ OR-gates, each with fan-in $2$.  It is easy to
  see that if we add such trees of OR-gates before $\layer{2}$, we
  increase the loss term in Equation~\eqref{eqn:approx_loss} to at
  most $m\log(m)/N$.  We can easily offset this by increasing $N$
  appropriately, and so retain our conclusion even when targeting sets
  have size at most $2$.
\end{proof}
\fi

\section{Experimental Results}
\newcommand{\D}{d^{-}}
\newcommand{\inD}{\Gamma^{+}}
\newcommand{\outD}{\Gamma^{-}}
\newcommand{\td}{td}
\newcommand{\hept}{\textsf{NetHEPT}}
\newcommand{\phy}{\textsf{NetPHY}}
\newcommand{\facebook}{\textsf{Facebook}}
\newcommand{\amazon}{\textsf{Amazon}}
\newcommand{\emailnet}{\textsf{EnronEmail}}
\newcommand{\opinion}{\textsf{Epinions}}

\newcommand{\degreeInt}{\textsf{DegreeInt}}
\newcommand{\degreeFrac}{\textsf{DegreeFrac}}
\newcommand{\discountFrac}{\textsf{DiscountFrac}}
\newcommand{\discountInt}{\textsf{DiscountInt}}
\newcommand{\randomAlg}{\textsf{RandomInt}}
\newcommand{\uniformAlg}{\textsf{UniformFrac}}
\newcommand{\imw}{7.0cm}
\newcommand{\imh}{4.1cm}


{\bf Datasets. } 
We use the following real-world networks for evaluating our claims.
Table~\ref{table:net} gives some statistical information about these networks.

\begin{table}[b]
\centering
\begin{tabular}{c|r|r|c|c}
Network    & $\#$ nodes &  $\#$ edges  & Avg.\ deg.  & Directed\\
\hline
\hept\ & 15,233 & 58,891 & \phantom{0}7.73  & No\\ 
\phy\ & 37,154 & 231,584 & 12.46 & No\\
\facebook\ & 4,039 & 88,234 & 21.84 & No\\ 
\amazon\ & 262,111 & 1,234,877 & \phantom{0}4.71 & Yes\\ 
\end{tabular}
\caption{Information about the real-world networks we use.}
\label{table:net}
\end{table}

\begin{itemize}
\item {\hept:} An academic collaboration network based on ``High Energy Physics --- Theory'' section of the e-print arXiv\footnote{\url{http://www.arXiv.org}} with papers from 1991 to 2003.
In this network, nodes represent authors and edges represent co-authorship relationships. This network is available at \url{http://research.microsoft.com/en-us/people/weic/graphdata.zip}.	

\item {\phy:} Another academic collaboration network, taken from the full
``Physics'' section of the e-print arXiv.
Again, nodes represent authors and edges represent co-authorship relationships. The network is available at \url{http://research.microsoft.com/en-us/people/weic/graphdata.zip}.	

\item {\facebook:} A surveyed portion of the Facebook friend network.
The nodes are anonymized Facebook users and edges represents friendship relationships. The data is available at \url{http://snap.stanford.edu/data/egonets-Facebook.html}.


\item {\amazon:} Produced by crawling the Amazon website based on the following observation: customers who bought product $i$ also bought product $j$. In this network, nodes represent products and there is a directed edge from node $i$ to node $j$ if product $i$ is frequently co-purchased with product $j$. This network is based on Amazon data in March 2003. The data is available at \url{http://snap.stanford.edu/data/amazon0302.html}.

\end{itemize}

\begin{figure*}[!ht]
\centering
\begin{subfigure}[b]{0.4\textwidth}
\includegraphics[width=\imw, height=\imh]{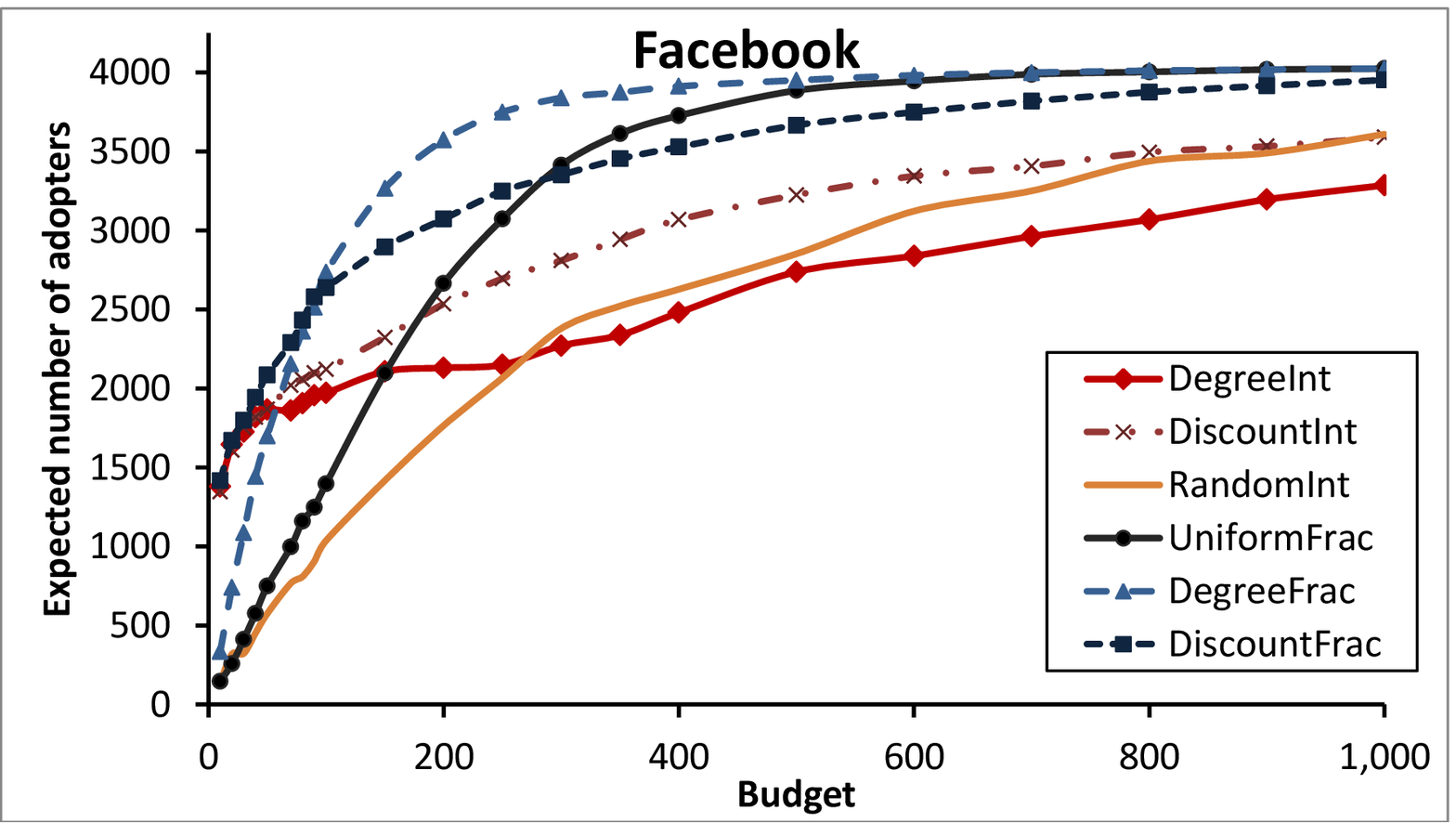} 
\label{fig:facebook_LT}
\end{subfigure}
~
\begin{subfigure}[b]{0.4\textwidth}
\includegraphics[width=\imw, height=\imh]{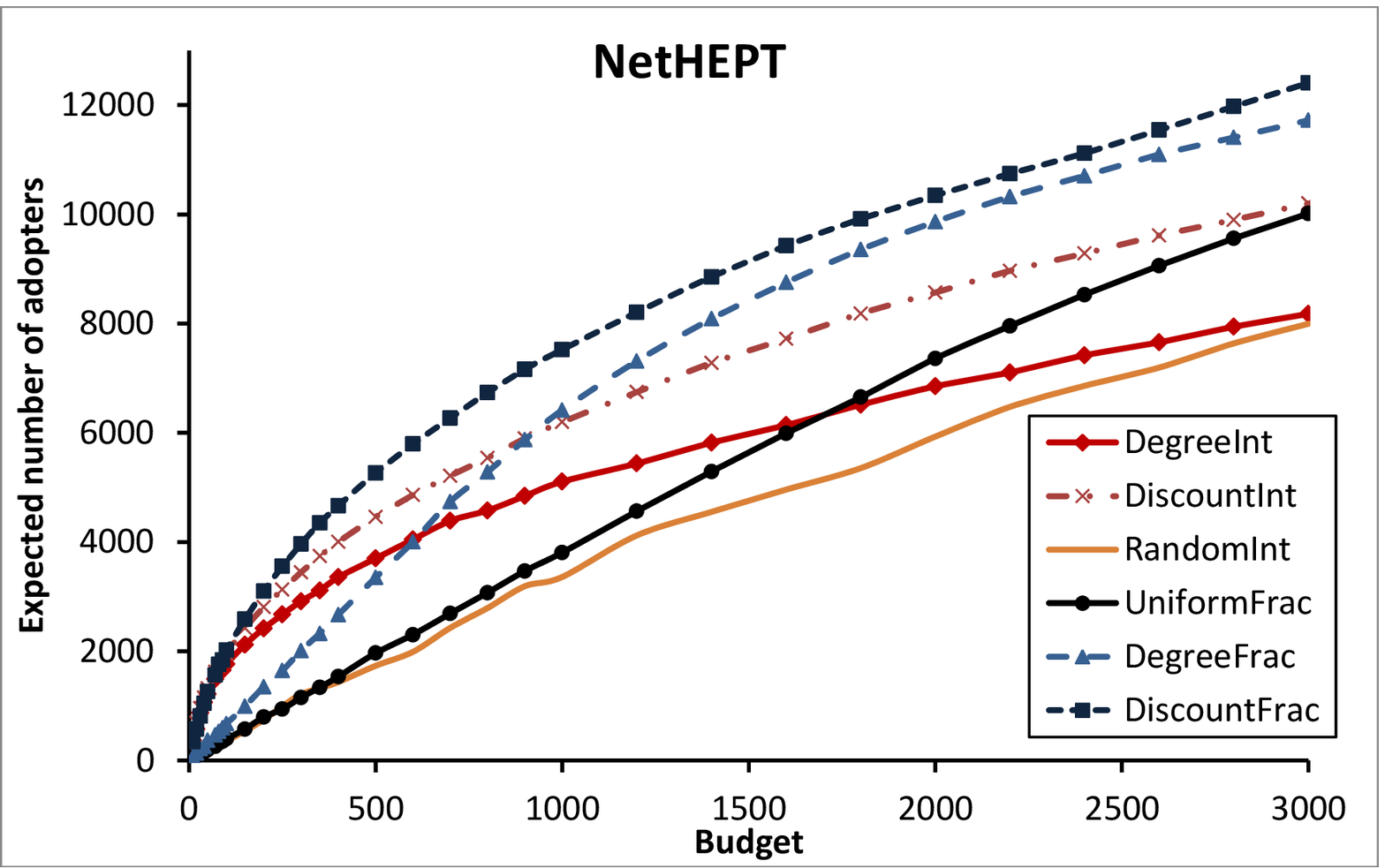} 
\label{fig:hept_LT}
\end{subfigure}
~
\begin{subfigure}[b]{0.4\textwidth}
\includegraphics[width=\imw, height=\imh]{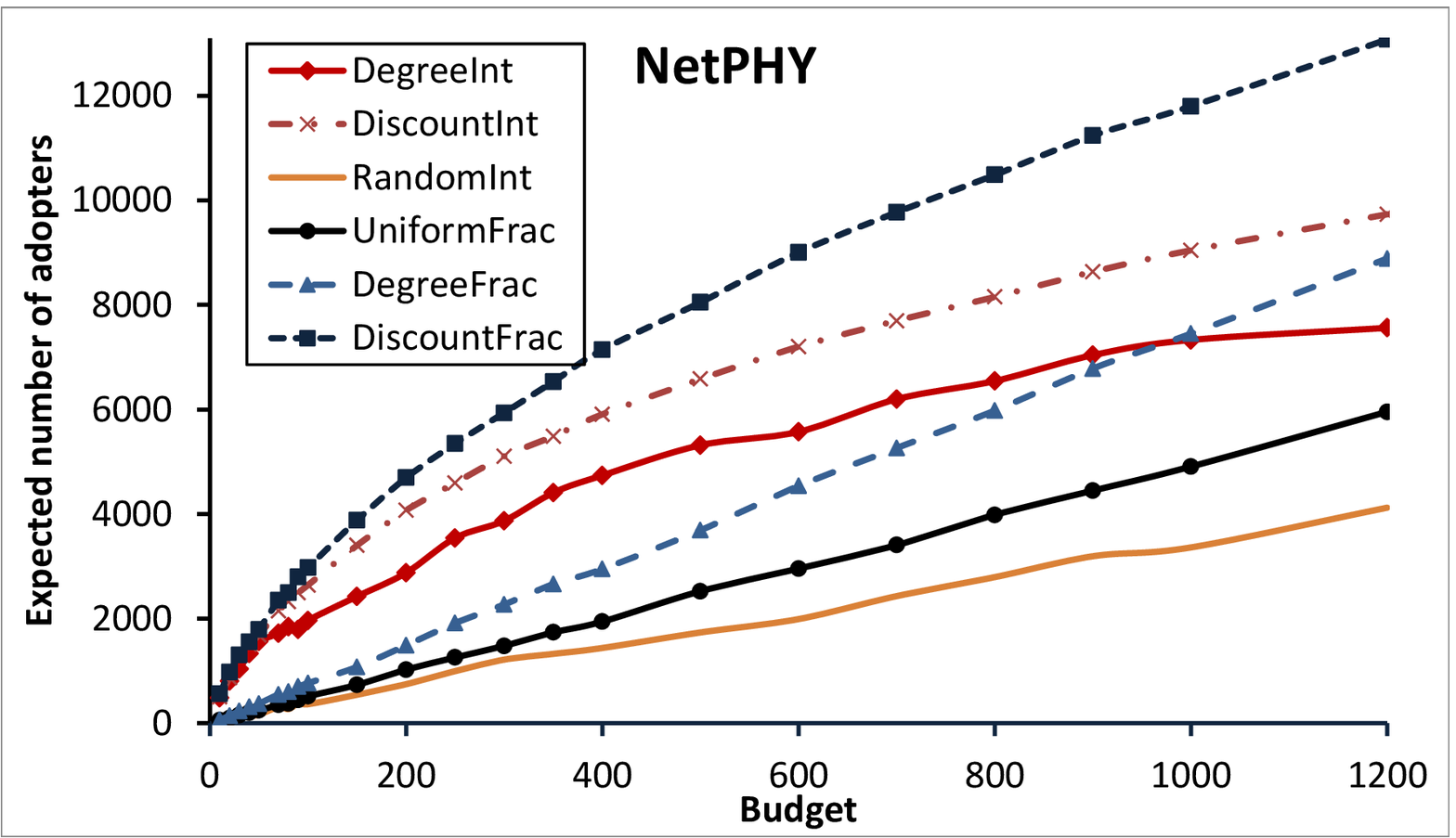} 
\label{fig:phy_LT}
\end{subfigure}
~
\begin{subfigure}[b]{0.4\textwidth}
\includegraphics[width=\imw, height=\imh]{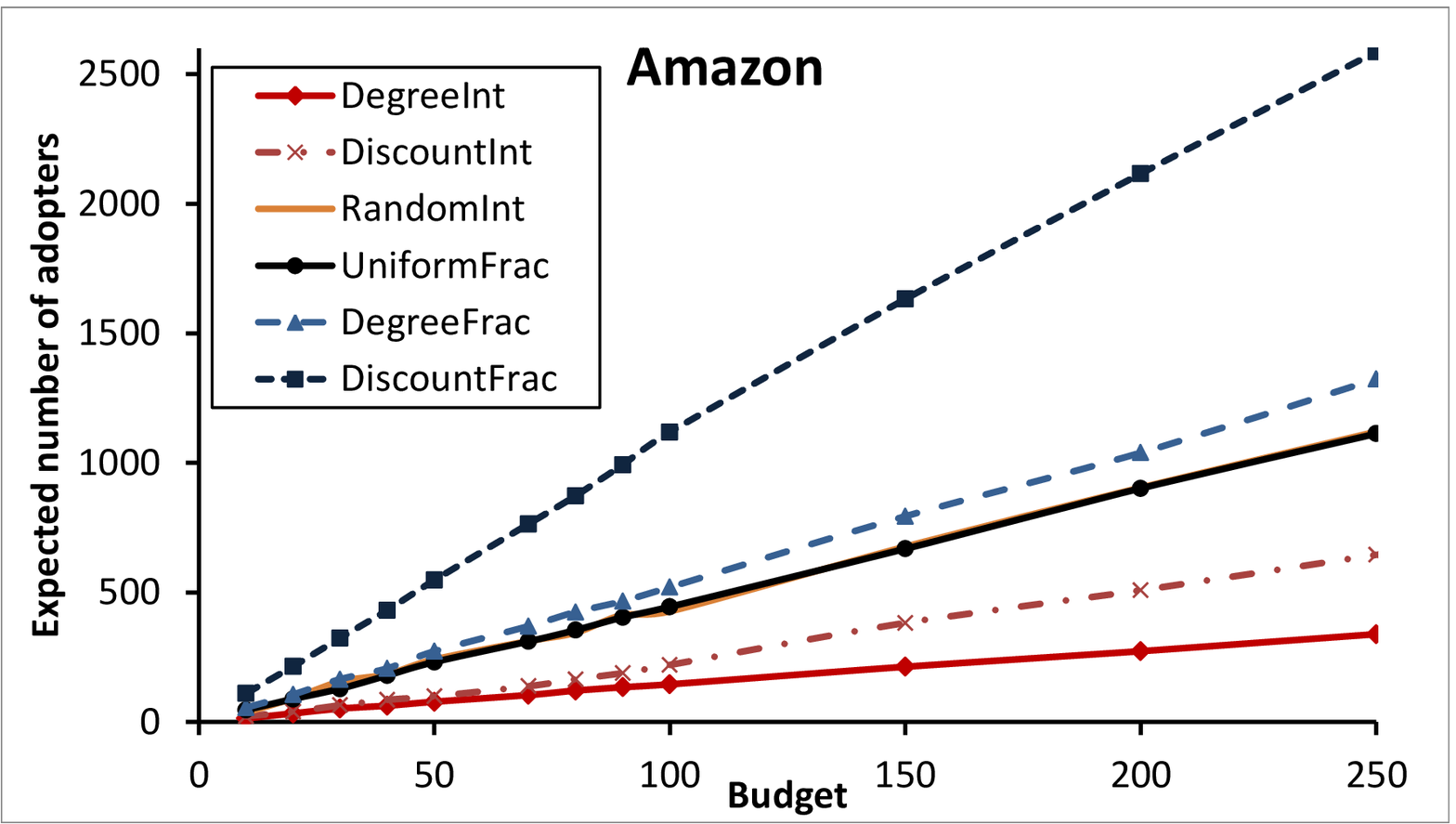} 
\label{fig:amazon_LT}
\end{subfigure}
\caption{Performance of different algorithms on \facebook, \hept, \phy, and \amazon. The weights of edges are defined based on the weighted cascade model. The $x$ axis is the budget and the $y$ axis is the expected number of adopters.}
\label{fig:LT}
\end{figure*}

\begin{figure*}[!ht]
\centering
\begin{subfigure}[b]{0.4\textwidth}
\includegraphics[width=\imw, height=\imh]{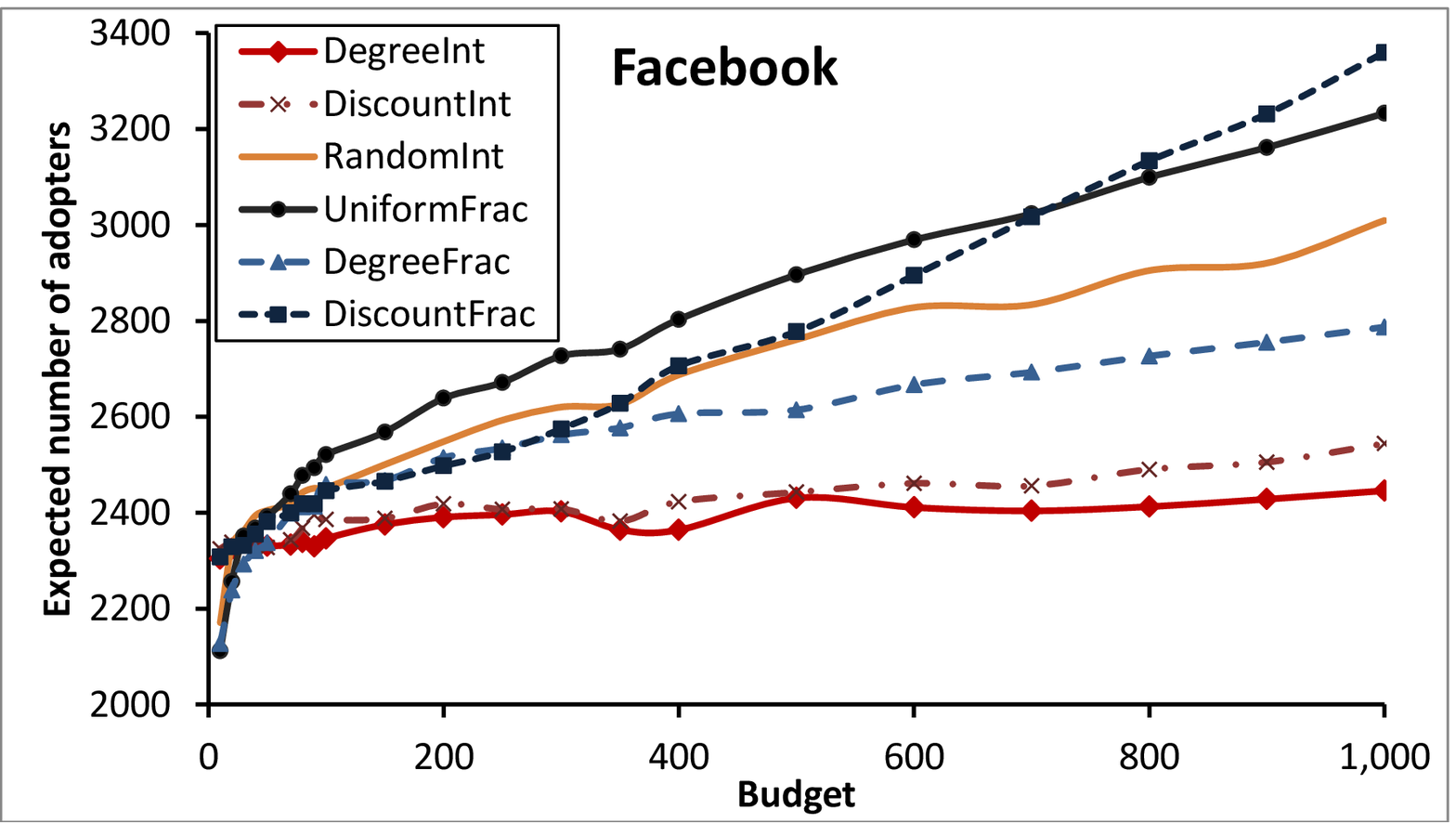} 
\label{fig:facebook_TR0}
\end{subfigure}
~
\begin{subfigure}[b]{0.4\textwidth}
\includegraphics[width=\imw, height=\imh]{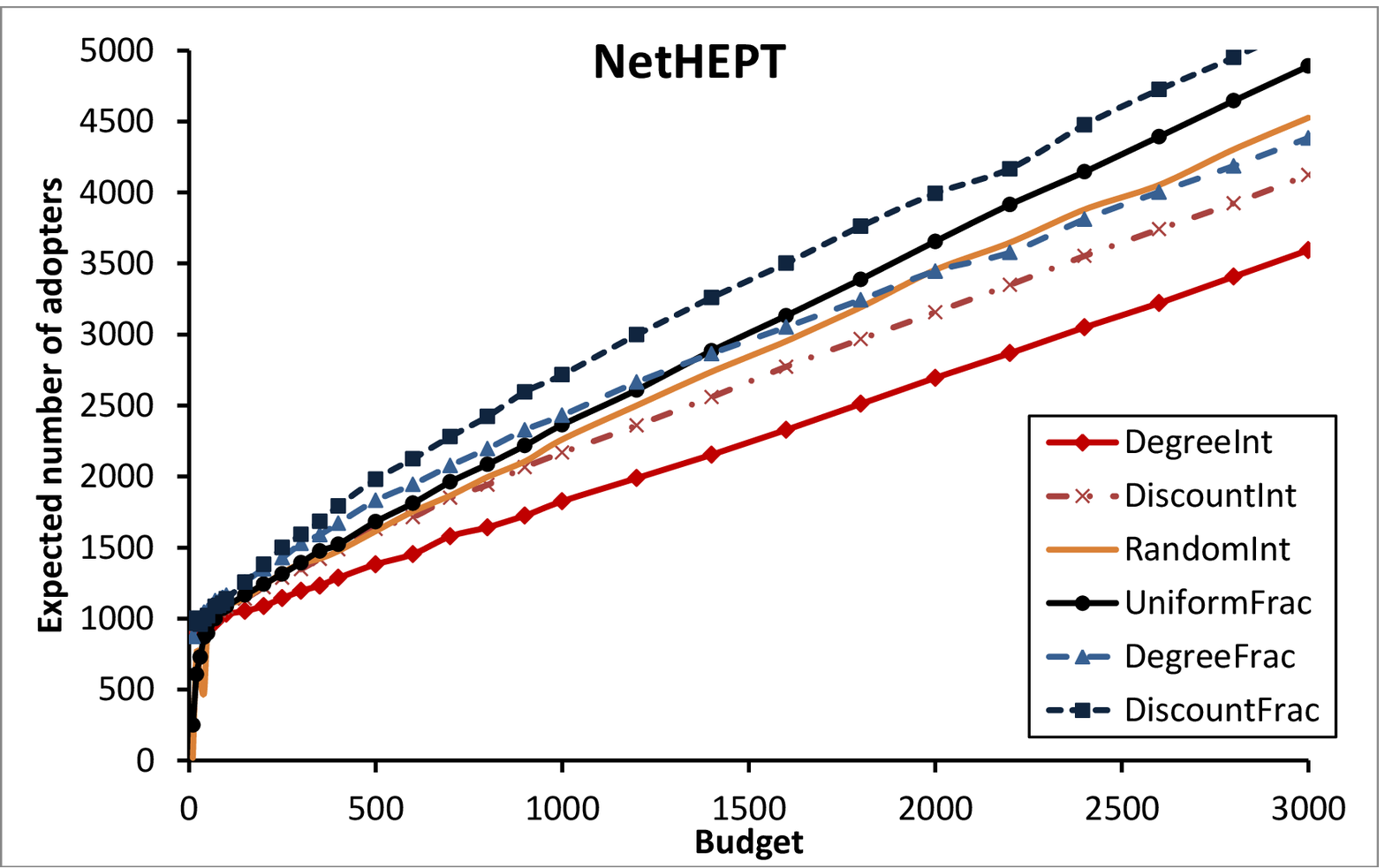} 
\label{fig:hept_TR0}
\end{subfigure}
~
\begin{subfigure}[b]{0.4\textwidth}
\includegraphics[width=\imw, height=\imh]{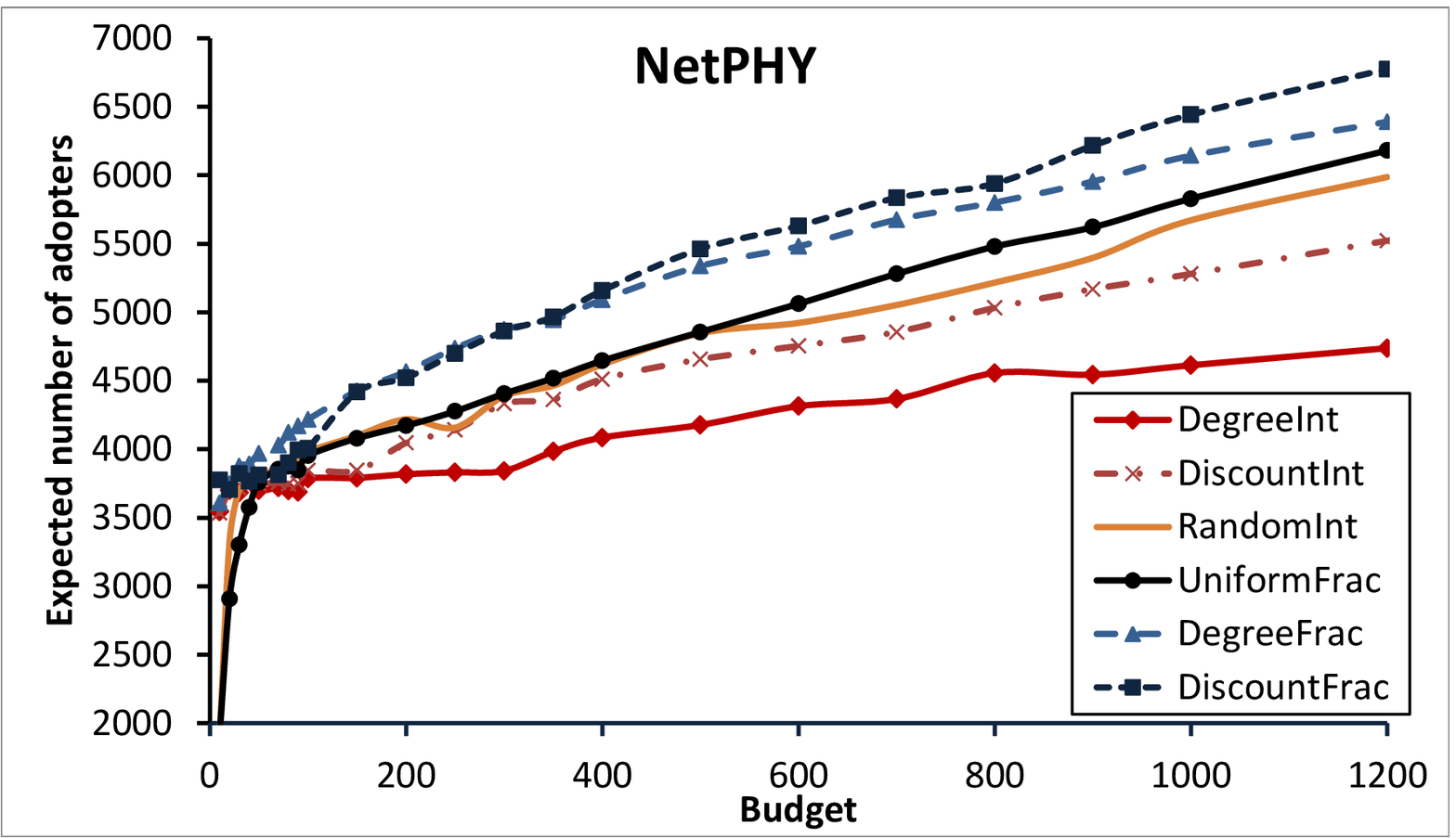} 
\label{fig:phy_TR0}
\end{subfigure}
~
\begin{subfigure}[b]{0.4\textwidth}
\includegraphics[width=\imw, height=\imh]{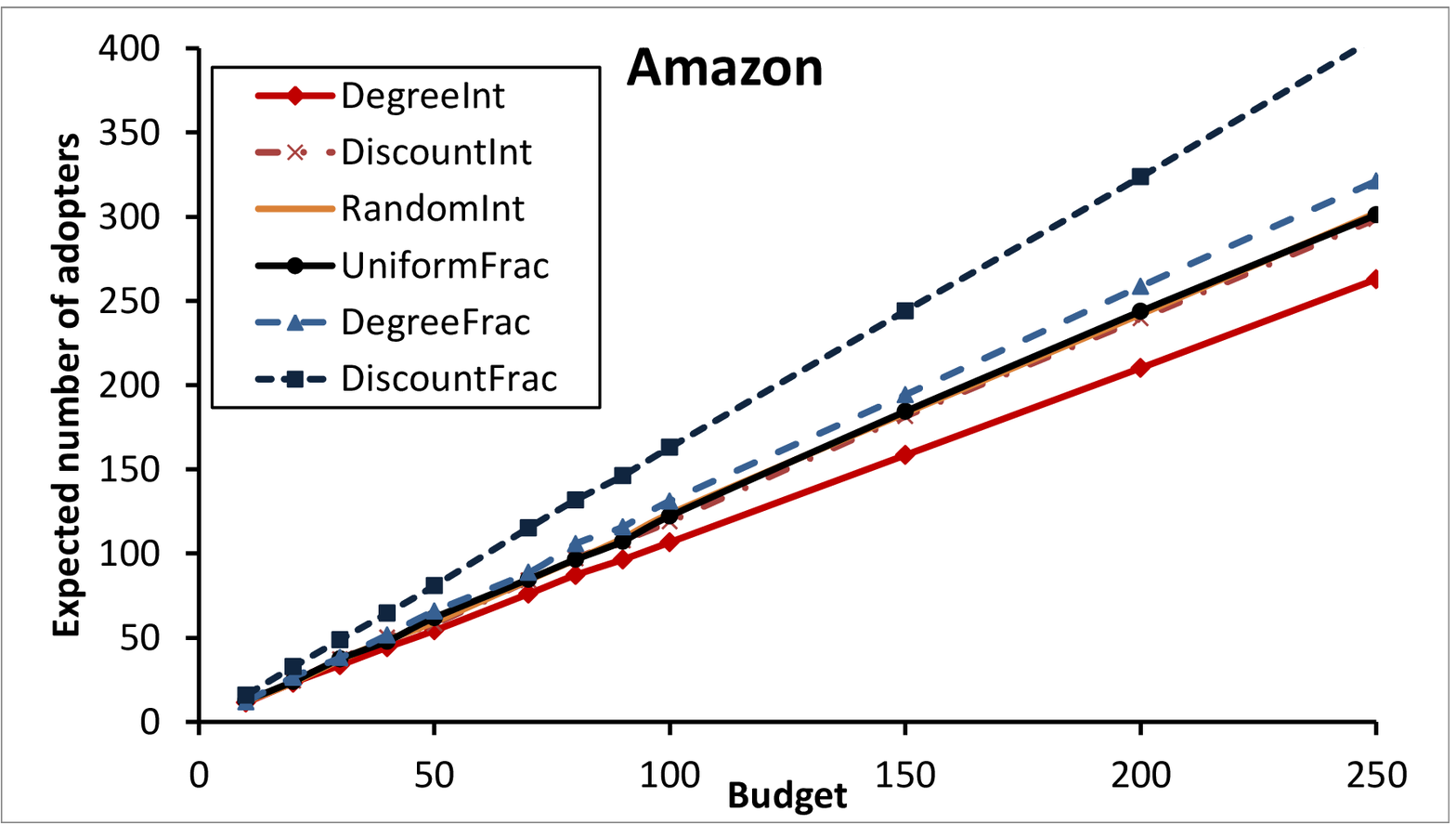} 
\label{fig:amazon_TR0}
\end{subfigure}
\caption{Performance of different algorithms on \facebook, \hept, \phy, and \amazon. The weights of edges are defined based on the TRIVALENCY model. The $x$ axis is the budget and the $y$ axis is the expected number of adopters.}
\label{fig:TR0}
\end{figure*}

{\bf Algorithms. }
We compare the following algorithms in this study. The first three algorithms are for the integral influence model, and the last three algorithms work for the fractional influence model.

\begin{itemize}
\item {\degreeInt:} A simple greedy algorithm which selects nodes with the largest degrees. This method was used by Kempe et al.~\cite{kkt03} and Chen et al.~\cite{Chen2009a} as well.

\item {\discountInt:} A variant of \degreeInt\ which selects node $u$ with the highest degree in each step. Moreover, after adding node $u$ to the seed set, the algorithm decreases the degrees of neighbors of $u$ by~$1$. This method was proposed and evaluated by Chen et al.~\cite{Chen2009a}.


\item {\randomAlg:} This algorithm randomly adds $B$ nodes to the seed set, i.e., by spending $1$ on each of them. We use this algorithm as a baseline in our comparisons. Other works \cite{kkt03, Chen2009a, Chen2010} also use this algorithm as a baseline.

\item {\degreeFrac:} This algorithm selects each node fractionally proportional to its degree. In particular, this algorithm spends 
on node $i$ where $B$ is the budget, $\D_i$ is the out-degree of node $i$, and $m$ is the total number of edges\footnote{If the graph is undirected, the cost is $2m$ instead of~$m$}.

\item {\discountFrac:}  A heuristic for the fractional case given by Algorithm~\ref{Alg:frac}. Let $\outD_v(A)$ be the total sum of the weight of edges from node $v$ to set $A$, and $\inD_v(A)$ be the total sum of the weight of edges from set $A$ to node $v$. 
This algorithm starts with an empty seed set $S$, and in each step it adds node $v \not \in S$ with the maximum $\outD_v(V-S)$ to seed set $S$ by spending $\max\{0, 1- \inD_v(S)\}$ on node $v$. Note that in each step the total influence from the current seed set $S$ to node $v$ is $\inD_v(S)$, and it is enough to spend $1- \inD_v(S)$ for adding node $v$ to the current seed set $S$. Note that no node would pay a positive amount, and the algorithm spends $\max\{0, 1- \inD_v(S)\}$ on node $v$.

\item {\uniformAlg:} This algorithm distributes the budget equally among all nodes. We use this algorithm as another baseline in our comparisons.
\end{itemize}


\begin{algorithm}
        \textbf{Input:} Graph $G=(V, E)$ and budget $B$ \\
        \textbf{Output:} Influencing vector $\mathbf{x}$ \\
\begin{algorithmic}[1]
\State $S \gets 0, b \gets B, \mathbf{x} \gets \mathbf{0}$
\While {$b > 0$}
\State $u \gets \argmax_{v \in V-S} \{\outD_v(V-S)\}$
\State $x_u \gets \min \{b, \max\{0, 1 - \inD_u(S)\}\}$
\State $b \gets b - x_u$
\State $S \gets S \cup \{ u \}$
\EndWhile
\State \Return $\mathbf{x}$
\end{algorithmic}
\caption{\discountFrac}
\label{Alg:frac}
\end{algorithm}

All these heuristic algorithms are fast and are designed for running on large real-world networks. In particular, algorithms \degreeInt\ and \degreeFrac\ only need the degree of nodes. 
We can use a Fibonacci heap to implement \discountInt, resulting in a running time of $O(B \log n +m )$.
Similarly, the running time of \discountFrac\ is $O(n \log n +m)$ using a Fibonacci heap.\footnote{In \discountFrac, the while loop (lines 4--9 of Algorithm~\ref{Alg:frac}) may run for $n$ steps even when budget $B$ is less than $n$. Hence, the running time is $O(\max\{n,B\} \log n +m)=O(n \log n+m)$ instead of $O(B \log n +m)$.}
Algorithms \randomAlg\ and \uniformAlg\ are linear-time algorithms.
It also has been shown that the performance of \discountInt\ almost
matches the performance of the greedy algorithm which maximizes a
submodular function \cite{Chen2009a}. Since the greedy algorithm
becomes prohibitively expensive to run for large networks, this
motivates us to use \discountInt\ as a reasonable benchmark for
evaluating the power of the integral influence model.

{\bf Results. }
We have implemented all algorithms in C++, and have run all experiences on a server with two 6-core/12-thread 3.46GHz Intel Xeon X5690 CPUs,
with 48GB 1.3GHz RAM. We run all of the aforementioned algorithms for finding the activation vector/set, and compute the performance of each algorithm by running 10,000 simulations and taking the average of the number of adopters.

We first examine the performance of a fractional activation vector in the {\em weighted cascade model}, where the weight of the edge from $u$ to $v$ is $\frac{1}{\D_v}$, where $\D_v$ is the in-degree of node~$v$. Note, the total weight of incoming edges of each node is $\sum_{uv}{w_{uv}}=\sum_{uv}\frac{1}{\D_v} = 1$. This model was proposed by Kempe et al.~\cite{kkt03}, and it has been used in the literature~\cite{Chen2010, Chen2009a, Chen2010a}.
See Figure \ref{fig:LT} for results.

We then compare the performance of various algorithms when the weight of edges are determined by the \emph{TRIVALENCY} model,
in which the weight of each edge is chosen uniformly at random from the set $\{0.001, 0.01, 0.1\}$.  Here $0.001$, $0.01$, and $0.1$ represent low, medium, and high influences.  In this model, the total sum of the weights of incoming edges of each node may be greater than~$1$.
This model and its variants have been used in \cite{kkt03,Chen2009a,Chen2010}.
We run all proposed algorithms on real-world networks when their weights are defined by  TRIVALENCY model. See Figure \ref{fig:TR0} for results. 


{\bf Discussion. }
In most of the plots, algorithms for the fractional 
influence model do substantially better than algorithms for the integral influence model. Overall, 
for most datasets, \discountFrac\ is the best algorithm, with the only exception being 
the Facebook dataset. As a simple metric of the power of the fractional model versus the 
integral model, we consider the pointwise performance gain of fractional model 
algorithms versus the integral model algorithms. i.e., for a given budget, we 
compute the ratio of expected number of adopters for the fractional model with 
the most adopters and the expected  number of adopters for the integral model algorithm 
with the most adopters. Depending on the dataset, we get a mean pointwise performace gain
between $3.4$\% (Facebook dataset, TRIVALENCY model) and $142.7$\% (Amazon 
dataset, weighted cascade model) with the mean being $31.5$\% and the median being 
$15.7$\% over all the datasets and both models (weighted cascade and TRIVALENCY).
Among the heuristics presented for the integral model, \discountInt\ is probably the
best. If we compare just it to its fractional adaptation, \discountFrac, we get a similar picture:
the range of average performace gain is 
between $9.1$\% (Facebook, TRIVALENT model) and $397.6$\% (Amazon, weighted cascade
 model) with a mean of $64.1$\% and a  median of $15.6$\%. 

In summary, the experimental results clearly demonstrate that the fractional model leads to a significantly higher number of adopters across a wide range of budgets on diverse datasets.


\bibliographystyle{plainnat}
\bibliography{fractionalkkt}  

\end{document}